\documentclass[fleqn,12pt]{wlscirep}
\usepackage[utf8]{inputenc}
\usepackage[T1]{fontenc}
\usepackage{cases,soul,booktabs}
\usepackage{amsthm}
\usepackage{lineno}
\usepackage{algpseudocode}
\usepackage{algorithm}
\usepackage{geometry}
\usepackage{array}
\usepackage{booktabs}

\newcounter{algsubstate}


\newtheorem{theorem}{Result}
\newtheorem{lemma}{Lemma}

\newcommand{\II}{\mathrm{I}}
\newcommand{\diag}{\mathrm{diag}}
\newcommand{\xx}{\mathbf{x}}
\newcommand{\yy}{\mathbf{y}}

\newcommand{\vv}{\mathbf{v}}

\newcommand{\zz}{\mathbf{z}}

\newcommand{\ddd}{\mathrm{d}}
\newcommand{\ZZ}{\mathbf{Z}}
\newcommand{\Am}{\mathbf{A}}
\newcommand{\g}{\mathbf{g}}
\newcommand{\HH}{\mathbf{H}}
\newcommand{\R}{\mathbb{R}}

\newcommand{\one}{\mathbf{1}}
\newcommand{\onek}{\mathbf{1_K}}
\newcommand{\Q}{\mathbf{Q}}
\newcommand{\cc}{\mathbf{c}}
\newcommand{\aaa}{\mathbf{a}}
\newcommand{\bb}{\mathbf{b}}
\newcommand{\nn}{\mathbf{n}}
\newcommand{\YY}{\mathbf{Y}}

\newcommand{\bbeta}{\mathbf{\beta}}
\newcommand{\aalpha}{\mathbf{\alpha}}

\title{
Quantum Annealing Continuous Optimisation in Renewable Energy
}

\author[1,*]{Mansour T.A. Sharabiani}
\author[2]{Vibe B. Jakobsen}
\author[2]{Martin Jeppesen}
\author[3]{Alireza S. Mahani}
\affil[1]{Imperial College of Science, Technology and Medicine, London, UK}
\affil[2]{Nature Energy, Odense, Denmark}
\affil[3]{Davidson Kempner Capital Management, New York, USA}

\affil[*]{mansour.taghavi-azar-sharabiani05@imperial.ac.uk}

\begin{abstract}
Renewable energy optimisation poses computationally-intensive challenges. Yet, often the continuous nature of the decision space precludes the use of many emerging, non-von-Neumann computing platforms such as quantum annealing, which are limited to discrete problems. We propose Quantum Annealing Continuous Optimisation (QuAnCO), a Trust Region (TR)-based algorithm, where the TR Newton sub-problem is transformed into Quadratic Unconstrained Binary Optimisation (QUBO), thereby allowing the use of Ising solvers such as D-Wave's quantum annealer. This transformation to QUBO is done by 1) using a hyper-rectangular shape for the TR, 2) discrete representation of each continuous dimension using an interval-bounded integer, and 3) binary encoding of the resulting bounded integers. We tackle a real-world challenge of optimising the biomass mix selection for Nature Energy, the largest biogas producer in Europe, thus providing evidence of feasibility and performance advantage in using QuAnCO in green energy production, and beyond.  
\end{abstract}

\begin{document}

\flushbottom
\maketitle

\thispagestyle{empty}

%\noindent Please note: Abbreviations should be introduced at the first mention in the main text – no abbreviations lists. Suggested structure of main text (not enforced) is provided below.

\section{Introduction}\label{sec-introduction}

The biogas industry - producing biomethane using anaerobic digestion (AD) of organic waste - has become an important partner in the global campaign against climate change~\cite{jain2019biogas}, contributing to greenhouse gas reduction via multiple pathways such as replacing fossil fuels, avoiding methane slips from manure, storing carbon in soils, producing green fertilisers and enabling carbon re-use. For a large, multinational biogas producer such as Nature Energy (NE), daily operations involve a multitude of decisions about sourcing, distribution, pre-tank storage and mixing of biomass in co-digestion tanks. Optimising such complex decisions would make a significant impact on the operational efficiency of the biogas plants and hence the long-term economic sustainability of the industry. However, the combinatorial explosion of the decision space in such high-dimensional optimisation problems makes them computationally prohibitive, often exceeding the limits of classical computing platforms.

While research into non-von-Neumann computing frameworks - including Quantum Annealing (QA) - is an active and promising avenue, yet there remains a significant limitation in their applications: QA targets discrete (or combinatorial) optimisation problems, whereas many problems in industrial optimisation, especially those facing renewable energy producers such as NE, are set in continuous spaces. In this paper, we address this mismatch by proposing a novel optimisation algorithm, which we call the Quantum Annealing Continuous Optimisation (QuAnCO).%, where the constrained quadratic optimisation sub-problem in Trust Region Newton (TRN) is transformed into a Quadratic Unconstrained Binary Optimisation (QUBO) problem, thus allowing the use of Ising solvers such as D-Wave's QA devices. This transformation is achieved by using a bounded integer - and then a binary vector - to represent each continuous dimension, and using a quadratic `model function' to approximate the underlying, nonlinear cost function.

%using a hyper-rectangular shape for the trust region (TR), and subsequently using a bounded integer - and then a binary vector - to represent each dimension of the resulting TR.

\textbf{Biomass Selection Optimisation} \quad Much research in biogas production optimisation has been focused on training `black-box' models on the reactor yield data, and using the resulting prediction function in a derivative-free optimisation algorithm~\cite{jha2017renewable}. For example, refs~\cite{qdais2010modeling,kana2012modeling} use Artificial Neural Networks (ANN) for training a black-box model, and Genetic Algorithms (GA) for optimisation. Ref~\cite{jacob2016modeling} compares ANN+GA against response surface methodology using central composite design for data generation.

%The industry-standard, white-box model of anaerobic digestion (AD) is ADM1~\cite{batstone2002iwa} and its modified versions~\cite{batstone2006review}, using structured equations to represent the biochemical and physicochemical processes involved in AD, resulting in 26-32 state variables. Full calibration of ADM1 in practical settings with limited data is often a challenge~\cite{hagos2017anaerobic}, motivating research into the so-called `grey-box' models that combine structured equations with experimental data to create parsimonious representations of AD and biogas production in industrial settings.

Biomethane Potential (BMP) experiments offer a practical path to `grey-box' models, where measurements collected in small batch reactors are used to predict biomethane yield in industrial-scale, continuous stirred-tank reactors, sometimes via parameter calibration of `white-box' models of anaerobic digestion (AD) such as as ADM1~\cite{batstone2002iwa,batstone2006review}. Over the years, many parametric forms have been tested and/or recommended for fitting BMP data, including first-order kinetics (or exponential)~\cite{pitt1999use}, modified Gompertz~\cite{zahan2018anaerobic}, cone~\cite{el2013kinetics}, Fitzhugh~\cite{cao2017methane}, and transfer function~\cite{li2015comparison}, among others. (For a recent review, see ref~\cite{pererva2020existing}.) More consistency in execution and reporting of BMP experiments are needed to accurately compare the results and reach consensus on the best model(s) to use under different circumstances~\cite{holliger2016towards,ohemeng2019perspectives,pererva2020existing}. In this paper, we have selected three models for our experiments: exponential, cone, and Cauchy~\cite{pitt1999use}. See Section \ref{subsec-methods-biomass-problem} for details.

\textbf{Nonlinear Optimisation Algorithms} \quad Algorithms for unconstrained optimisation of smooth functions broadly fall under two categories: line search and trust region (TR)~\cite{nocedal2006numerical}. In line search, the movement direction is chosen before the step length. Two important members of this category are BFGS and Conjugate Gradient (CG) algorithms~\cite{nocedal2006numerical}. In contrast, during TR methods the maximum step length is fixed before searching for the direction. TR strategies are considered reliable and robust, and can be applied to non-convex and ill-conditioned problems. It is also easier to establish convergence results for TR algorithms~\cite{yuan2015recent}. %TRN uses a second-order Taylor-series expansion to generate the model function (i.e., using the exact Hessian), and solves the resulting, constrained quadratic optimisation problem within a hyper-sphere surrounding the current point to generate a potential move. The observed vs. expected improvement in the cost function determines whether the move is accepted or rejected, and whether the TR should be expanded or contracted.

Trust Region Newton (TRN) - where the model function is a second-order, Taylor-series fit to the cost function - is a natural candidate to adapt for using Ising solvers: 1) the boundedness of the TR (with some modification) permits its representation using bounded integers, 2) the global nature of the algorithm means finding an exact solution in each iteration is not important and thus approximate - e.g., Ising - solvers can be quite effective~[ref~\cite{nocedal2006numerical}, Chapter~4], 3) the quadratic model function can be mapped to the Ising energy with pairwise coupling terms, and 4) the computationally critical component of TRN is solving the TR sub-problem, which is precisely what our proposed QuAnCO algorithm offloads to Ising solvers such as QA.

\textbf{Ising/QUBO Solvers} \quad In recent years, there has been an increased volume of research in alternative computing approaches. An important category is Ising or QUBO solvers, which aim to find the ground state of a coupled, binary system with second-order interactions. One such Ising solver is D-Wave's QA, based on superconducting circuits coupled via Josephson junctions~\cite{johnson2011quantum}. In QA, quantum tunneling is the mechanism used for inducing a non-zero probability for transitioning to higher energy states and escaping local minima. In Simulated Annealing (SA), this effect is achieved via thermal fluctuations~\cite{isakov2015optimised}.

Beyond QA, there are other hardware- and software-based Ising solvers, some referred to as `quantum-inspired', in various stages of research, development and commercialisation. Examples include Hitachi's CMOS annealer~\cite{yamaoka201520k}, Toshiba's Simulated Bifurcation Machine~\cite{goto2019combinatorial}, Fujitsu's digital annealer~\cite{aramon2019physics}, photon-based approaches~\cite{yamamoto2017coherent,pierangeli2019large,roques2020heuristic}, Bose-Einstein condensation~\cite{byrnes2011accelerated}, and FPGA-based Restricted Boltzmann Machines~\cite{patel2020ising}. Such research has also motivated building of high-performance and/or large-scale, classical Ising solvers~\cite{alphaqubo,isakov2015optimised}. Beyond Ising solvers, other notable examples of alternative computing frameworks proposed include neuromorphic computing~\cite{merolla2014million,cai2020power}, and atomic Boltzmann machine~\cite{kiraly2021atomic}. Factors driving the development and adoption of such hardware platforms include speed, connection density, size, data transfer overhead, cost, and energy efficiency.

Active research into new hardware and software for Ising solvers has, in turn, motivated research into finding new applications for Ising solvers~\cite{kochenberger2014unconstrained}, thus creating a virtuous cycle. By proposing our QuAnCO algorithm, we broaden the application of Ising solvers beyond discrete optimisation to include continuous optimisation problems, thus reinforcing the above-mentioned cycle.

\section{Results}\label{sec-results}

\subsection{Bound-Constrained TRN}\label{subsec-trn-bound}

Consider the unconstrained optimisation problem
\begin{equation}
\min_{\xx \in \R^K} f(\xx)     
\end{equation}
where $f(\xx)$ is a twice-differentiable, non-convex cost function. TR algorithms work by iteratively minimising a model function within a neighborhood surrounding the current point, $\xx_0$, called the `trust region'. The model function, $f^*(\xx)$, is often a quadratic approximation to the cost function:
\begin{equation}\label{eq-quad-approx}
    f^*(\xx) = f_0 + \g_0^\top \, (\xx - \xx_0) + \frac{1}{2} \, (\xx - \xx_0)^\top \, \HH_0 \, (\xx - \xx_0)
\end{equation}
where $f_0$ and $\g_0$ are the cost function and its gradient, respectively, evaluated at $\xx = \xx_0$. The matrix $\HH_0$, is either the exact Hessian of $f$ evaluated at $\xx = \xx_0$ (TRN), or an approximation of it (TR quasi-Newton). In each iteration of TR, the following `sub-problem' is solved:
\begin{equation}\label{eq-subproblem}
\min_{\xx} f^*(\xx) \quad s.t. \quad || \xx - \xx_0 ||_2 \leq r,
\end{equation}
where $||.||_2$ is the 2-norm operator, forcing $\xx$ to be inside a hyper-sphere of radius $r$ centered at $\xx_0$. Solving the constrained optimisation problem of Eq.~\ref{eq-subproblem} produces a proposed move, $\xx'$. If the actual improvement in the cost function, $f(\xx) - f(\xx')$, is close to expected improvement, $f^*(\xx) - f^*(\xx')$, the proposed move is accepted, and the TR radius ($r$) may be expanded. Too small of an actual improvement, on the other hand, leads to rejection of the proposed move and shrinking of the TR.
%Assuming that we are at $\xx = \xx_k$ and the proposed step is $\xx'$, below is a standard protocol in TR software for accepting/rejecting the proposed step and adjusting the trust region radius, using the ratio of actual-to-expected improvement ($\rho_k$):

The computational cost of TRN consists of two parts: 1) computing the cost function and its derivatives - especially the Hessian - in order to evaluate the model function according to Eq.~\ref{eq-quad-approx}, and 2) solving the sub-problem of Eq. \ref{eq-subproblem}. Component 1 is dominated by Hessian evaluation and - theoretically - scales like $O(K^2)$ (ignoring implementation effects such as cache-size limits), while the scaling of component 2 can be as bad as $O(K^3)$, e.g., driven by a matrix decomposition step. Furthermore, unlike component 2, component 1 is `embarrassingly parallelisable'~\cite{herlihy2020art}. In summary, `solving the TR sub-problem' is the computational bottleneck of TRN, and the focus of our proposed QuAnCO algorithm.

\textbf{Removing Bound Constraints} \quad The biomass selection optimisation problem (Section~\ref{subsec-biomass-problem}) includes bound constraints:
\begin{equation}\label{eq-box-optim}
\begin{aligned}
\min_{\xx} \quad & f(\xx) \\
\textrm{s.t.} \quad & a_k \leq x_k \leq b_k \quad \forall \, k \in \{1, \hdots, K\}
\end{aligned}
\end{equation}
This includes the general case of a box constraint as well as the special cases of lower bound only ($b_k = +\infty$) and upper bound only ($a_k = -\infty$). To remove bound constraints, we introduce a $K$ element-wise nonlinear functions, $\eta_k()$'s, such that each has a domain $\R$ and a range matching the bounds in Eq.~\ref{eq-box-optim}. The new - unconstrained - optimisation problem is:
\begin{equation}\label{eq-nlp-yspace}
\begin{aligned}
\min_{\mathbf{y} \in \R^K} \quad & F(\mathbf{y}) \equiv f(\eta(\mathbf{y})),
\end{aligned}
\end{equation}
where $\mathbf{\eta}(\mathbf{y})$ is a shorthand for $[ \eta_1(y_1) \,\, \eta_2(y_2) \,\, \hdots \,\, \eta_K(y_K)]$. Applying the chain rule of derivatives, we obtain the following expressions for the gradient ($g_F$) and Hessian ($H_F$) for the transformed cost function:
\begin{subnumcases}{}
    g_F(\mathbf{y}) & = $\mathbf{\eta}'(\mathbf{y}) \circ g_f(\mathbf{\eta}(\mathbf{y}))$ \label{eq-remove-box-constraints-a} \\
    H_F(\mathbf{y}) & = $\mathrm{diag} \left( \mathbf{\eta}''(\mathbf{y}) \circ \nabla f(\mathbf{\eta}(\mathbf{y})) \right) + \left( \mathbf{\eta}'(\mathbf{y}) \, \mathbf{\eta}'^T(\mathbf{y}) \right) \circ H_f(\mathbf{\eta}(\mathbf{y}))$ \label{eq-remove-box-constraints-b}
\end{subnumcases}

where $\circ$ represents the element-wise vector/matrix multiplication, and the derivatives in $\mathbf{\eta}'$ and $\mathbf{\eta}''$ are also interpreted element-wise. See Section \ref{sec-methods} for derivation of the above, as well as specific nonlinear functions used in our experiments.

\subsection{Quantum Annealing Continuous Optimisation (QuAnCO)}\label{subsec-results-qtro}

QuAnCO involves two key changes to TRN: 1) using a rectangular shape (infinity-norm) for the TR, rather than spherical ($2$-norm), followed by 2) transforming the optimisation sub-problem of Eq. \ref{eq-subproblem} into a QUBO and submitting it to an Ising solver.

A symmetric, hyper-rectangular TR centered on $\xx_0$ can be defined as
\begin{equation}
    %\mathbf{l} \leq \yy \leq \mathbf{u},
    ||(\xx - \xx_0) \oslash \mathbf{r}||_{\infty} \leq 1,
\end{equation}
where $\oslash$ is the element-wise division, and $\mathbf{r}$ is the vector representing the half-length of the rectangular TR. Within this TR, the continuous vector $\xx$ can be approximated using a vector of bounded integers ($\mathbf{n}$):
\begin{equation}\label{eq-x-n}
    \xx(\nn) = \left( \xx_0 - \mathbf{r} \right) + \mathbf{\delta} \, \circ \, \nn,
\end{equation}
where $\nn \in \{ 0,1,\dots,N \}^K$ and $\mathbf{\delta} \equiv \frac{2}{N} \mathbf{r}$. Next, we use a binary representation for each element $n_k$ in $\nn$:
\begin{equation}
    n_k = \bb^\top \zz_k
\end{equation}
where $\bb$ is the $M$-digit binary basis $\begin{bmatrix} 1 & 2 & \cdots & 2^{M-1} \end{bmatrix}^\top$, and $\zz_k$ is a binary vector of length $M$, representing the range of integers $[0,2^M-1]$. (As such, we require $N+1$ to be a power of two.) To represent the vector $\nn$, we stack $\zz_k$'s row-wise to form a $K \times M$ matrix $\ZZ$:
\begin{equation}\label{eq-n-Z}
    \nn = \ZZ \, \bb.
\end{equation}
Combining Eqs.~\ref{eq-x-n} and \ref{eq-n-Z}, and using the vectorisation or \textit{vec} trick (see Section~\ref{sec-methods}), we obtain:
\begin{subnumcases}{}%\label{eq-def-x-x0}
    \xx - \xx_0 & = $-\mathbf{r} + \Am \, \zz$ \label{eq-def-x-x0}
    %\\
    %\dd & = $\mathbf{l} - \xx_0$ \label{eq-def-d}
    \\
    \Am & = $\bb^\top \otimes \mathrm{diag}(\delta)$ \label{eq-def-A}
\end{subnumcases}
where $\zz$ is a binary vector of length $K \, M$ resulting from stacking columns of $\ZZ$, and $\otimes$ is the Kronecker product. Using the above substitution in Eq.~\ref{eq-quad-approx} leads to the following QUBO (see Section~\ref{sec-methods} for details):
\begin{equation}
    \min_{\zz \in \mathbb{B}^{K M}} \quad \zz^\top \, \Q \, \zz 
\end{equation}
with:
\begin{equation}\label{eq-def-Q}
        \Q = \frac{1}{2} \, \Am^\top \, \HH_0 \, \Am + \mathrm{diag} \left( \Am^\top (\g_0 - \HH_0 \, \mathbf{r}) \right)
\end{equation}
The computational cost of $\Q$ in Eq.~\ref{eq-def-Q} is dominated by the term $\Am^\top \, \HH_0 \, \Am$. Using Eq.~\ref{eq-aHa-2} (Section~\ref{sec-methods}), and noting that each element of $\Delta \, \HH_0 \, \Delta$ can be calculated in $O(1)$ thanks to $\Delta$ being diagonal, the cost of calculating $(\bb \, \bb^\top) \otimes (\Delta \, \HH_0 \, \Delta)$ is $O(K^2 \, M^2)$, which is polynomial in problem size, $K$. Furthermore, if we take advantage of $b_m = 2^{M-1}$, this can be reduced to $O(K^2 \, M)$, since the number of unique elements of $\bb \, \bb^\top$ is reduced from $O(M^2)$ to $O(M)$. In summary, computational time for calculating $\Q$ scales quadratically with problem size.

Denseness of $\Q$ is an important parameter for many Ising solvers, which motivates the following result (proof in Section~\ref{sec-methods}).
\begin{theorem}\label{theorem-denseness}
For a dense but otherwise arbitrary Hessian $\HH_0$, the coefficient matrix $\Q$ given by Eq.~\ref{eq-def-Q} is also dense.
\end{theorem}

Discretisation error is not a major concern for QuAnCO since TR algorithms do not require exact solutions to the sub-problem~\cite[Chapter~4]{nocedal2006numerical}. Furthermore, optimising a discretised version of a `smooth' function should produce results that are similar to optimising the exact function, where smoothness can be defined using Lipschitz continuity~\cite{cobzacs2019lipschitz}. In particular, if the first or second derivative of a function are absolute-bounded, true function minimum - potentially lying in-between the discrete-grid nodes - cannot be much lower than the observed minimum. The next result quantifies this notion (proof in Section~\ref{sec-methods}).

\begin{theorem}\label{theorem-discretisation}
Discretisation error resulting from solving the QUBO in Equation~\ref{eq-def-Q}, rather than the underlying quadratic function in continuous space, has an upper bound of $\frac{1}{8} \lambda \, \mathbf{\delta}^\top \mathbf{\delta}$, where $\lambda$ is the maximum of largest positive eigenvalue of $\HH_0$ (if there is one) and zero, and $\mathbf{\delta}$ is the grid resolution vector. \end{theorem}

This result can be used in selecting grid resolution ($\delta$) in each step, using a `natural' scale or resolution ($\Delta f$) for the cost function: $\delta = (8 \Delta f / \lambda)^{1/2}$ (assuming an isotropic grid).

Algorithm~\ref{alg-qtro} summarises QuAnCO. Note that we check for improvement in function value before accepting the proposal, because that the combination of discretisation and using an approximate Ising solver can lead to both the actual and the expected improvements being negative, thus forming a positive ratio, $\rho$. Also, the condition $|| \hat{\mathbf{p}}_k \oslash \mathbf{r}_k ||_{\infty} = 1$ checks whether the proposed step falls on the hyper-rectangular TR boundary.

\begin{algorithm}
\caption{Quantum Annealing Continuous Optimisation (QuAnCO). Inputs are $M$: number of bits used to represent each dimension; $\mathbf{r}_0$: vector of TR initial size; $\mathbf{r}_{max}$: vector of maximum allowable sizes for TR; $\xx_0$: starting value; $f$: cost function; $\epsilon_1,\epsilon_2$: threshold parameter for convergence test.}\label{alg-qtro}
\begin{algorithmic}
\Require $M, \, \mathbf{r}_0, \mathbf{r}_{max}, \, \xx_0, f, \epsilon_1, \epsilon_2$
\State $\mathrm{converged} \gets false$
\State $k \gets 0$
\While{!converged}
\State $\g_k \gets \nabla f(\xx_k)$, $\HH_0 \gets \mathbf{J}(\g(\xx_k))$
\State Solve the sub-problem:
\begin{enumerate}
    \item Calculate $\Q$ according to Eq.~\ref{eq-def-Q}.
    \item Solve the resulting QUBO to obtain a proposed move, $\hat{\zz}_k$ in binary space.
    \item Convert proposed move to continuous space: $\hat{\mathbf{p}_k} \gets - \mathbf{r} + \Am \, \hat{\zz}_k$.
\end{enumerate}
\State Calculate the improvement ratio: $\rho_k \gets \left( f(\xx_k + \hat{\mathbf{p}}_k) - f(\xx_k) \right) / \left( \g_k^\top \hat{\mathbf{p}_k} + \hat{\mathbf{p}}_k^\top \HH_k \hat{\mathbf{p}}_k \right)$.
\If{$\rho_k < 1/4 \,\, \mathrm{or} \,\, f(\xx_k + \hat{\mathbf{p}_k}) > f(\xx_k)$}
    \State $\xx_{k+1} \gets \xx_k$
    \State $\mathbf{r}_{k+1} \gets \mathbf{r}_k / 4$
\Else
    \State $\xx_{k+1} \gets \xx_k + \hat{\mathbf{p}}_k$
    \If{$\rho_k > 3/4 \,\, \mathrm{and} \,\, || \hat{\mathbf{p}}_k \oslash \mathbf{r}_k ||_{\infty} = 1$}
        \State $\mathbf{r}_{k+1} \gets \min(2 \mathbf{r}_k, \mathbf{r}_{max})$
    \EndIf
\EndIf
\State $\mathrm{converged} \gets |f(\xx_k + \hat{\mathbf{p}}_k) - f(\xx_k)| \leq \epsilon_1 \,\, OR \,\, |\g_k^\top \hat{\mathbf{p}_k} + \hat{\mathbf{p}}_k^\top \HH_k \hat{\mathbf{p}}_k| \leq \epsilon_2$
\State $k \gets k + 1$
\EndWhile
\end{algorithmic}
\end{algorithm}

\subsection{Biomass Selection Optimisation}\label{subsec-biomass-problem}

In a co-digestion reactor at an industrial-scale biogas plant, several types of biodegradable material - or biomass - are continuously fed into the reactor where, with the help of bacterial micro-organisms, they undergo a complex, multi-stage process known as anaerobic digestion (AD)~\cite{batstone2002iwa}, at the end of which green energy is produced in the form of biomethane. Biomasses have diverse attributes in terms of the amount and rate of methane production, procurement and transportation costs, and regulatory implications. Figure~\ref{fig-biomass-diversity} (panels A-C) illustrates some of this diversity, using data collected by Nature Energy, the largest producer of biogas in Europe.

\begin{figure}[!h]
\centering
\includegraphics[width=\linewidth]{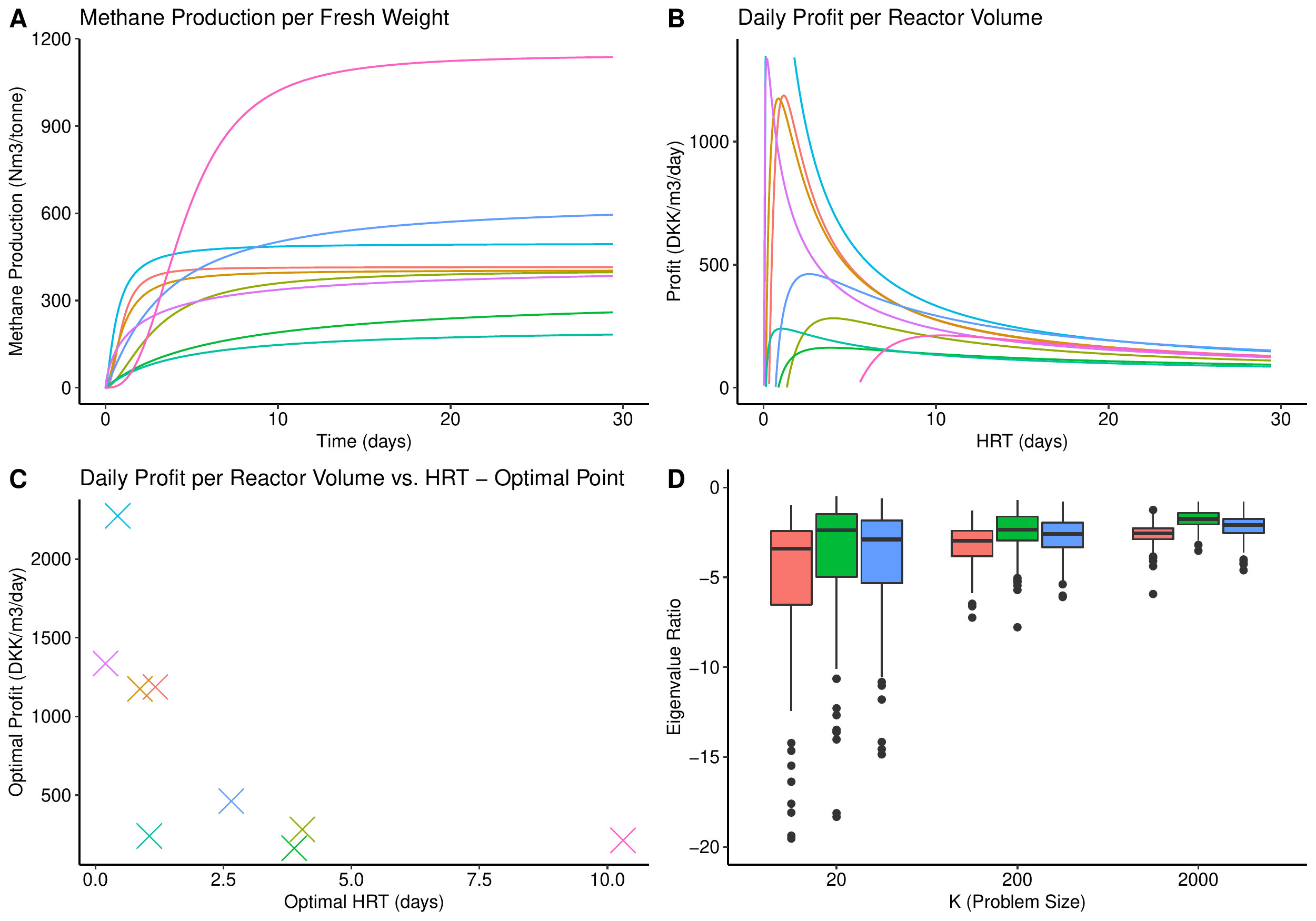}
\caption{Biomass selection optimisation for Nature Energy. \textbf{A}: Sample BMP curves (per fresh weight) biomasses used to generate synthetic data (Section~\ref{subsec-methods-biomass-problem}). \textbf{B}: Profitability (per reactor volume per day) vs. Hydraulic Retention Time (HRT) for the same group of biomasses. \textbf{C}: Scatter plot of maximum daily profit per reactor volume vs. optimal HRT for each biomass. \textbf{D}: Ratio of smallest to largest eigenvalue of the Hessian of the cost function. Each bar is based on 1000 random samples of biomasses ($N$ of 28). For each sample, the first element of $\xx_0$ is chosen to be the univariate optimal value for the first biomass, and the remaining elements are set to zero. Biomass properties have been collected and provided by Nature Energy. Note that the profit numbers (y-axis) in panels \textbf{B} and \textbf{C} have been adjusted with an offset to protect business sensitive information. green: cone, blue exponential, red: cauchy}
\label{fig-biomass-diversity}
\end{figure}

In this paper, we focus on the biomass selection problem in a single biogas reactor, where the goal is to set the daily rate for feeding each biomass into the reactor. We further assume that each type of biomass produces biomethane in the reactor according to the production curves fitted to lab data from biomethane potential (BMP) experiments. The time value used - across all biomass types - to read out the BMP data is the Hydraulic Retention Time (HRT) of the reactor, which is defined as the ratio of the active reactor volume over the total daily volume of biomass flowing through the reactor, and can be interpreted as the - approximate - average time each biomass spends in the reactor. The above assumptions lead to the following expression for the cost function, which is the total daily cost of sourcing the biomass, minus the daily revenue from selling the biomethane produced in the reactor:
\begin{equation}\label{eq-def-cost}
f(\xx) = \xx^\top \, \left(\cc - r \, \YY(X) \right), \quad x_k \geq 0, \quad \forall \, k = 1, \hdots, K.
%f(\xx) &=& - \cc^\top \, \xx + r \, \YY(\xx^\top \onek) ^\top \xx
\end{equation}
where $\xx$ is a vector of length $K$ representing daily flow volume of each biomass into the reactor, $X$ is the total daily feed, i.e., $X \equiv \sum_{k} x_k$, and $\YY \equiv \begin{bmatrix} Y_1(X) & Y_2(X) & \cdots & Y_K(X)\end{bmatrix}$, with $Y_k$'s representing parametric yield functions. (Note that $HRT = V / X$, where $V$ is the total active volume of the reactor, and hence yield functions could also be expressed in terms of reactor $HRT$.) 

The above cost function has two important properties. First, it is non-convex, which means there are many places in the parameter space where its Hessian has negative eigenvalues. Panel D of Figure~\ref{fig-biomass-diversity} illustrates this point. Optimisation of non-convex cost functions, even in the quadratic form, is a non-deterministic polynomial-time (NP)-hard problem~\cite{pardalos1991quadratic}. Second, the minima of the cost function of Eq.~\ref{eq-def-cost} can be identified mathematically, per below.

\begin{theorem}\label{theorem-cost-minima}
For cost function of form~\ref{eq-def-cost}, all minima - global and local - are single-ingredient, i.e., they must satisfy $\sum_k \II[x_k > 0] \leq 1$. $I()$ is the indicator function.
\end{theorem}

\begin{proof}
We use proof by contradiction. Consider a multi-ingredient point $\xx$, i.e., $\sum_k \II[x_k > 0] = m > 1$. Assume that $k_1$ is the ingredient with the largest value of $r \, Y_k(X) - c_k$ among the set of non-zero ingredients, $U = \{k_1,k_2,\hdots,k_m\}$. Consider a small step $\ddd \xx$ such that $d x_k = 0, k \notin U$ \& $d x_k < 0, k \in U, k \neq k_1$ \& $d x_{k_1} = - \sum_{k' \in U, k' \neq k_1} d x_{k'} > 0$. Such a step keeps $X$ constant, i.e., $dX = \sum_k d x_k = 0$. Furthermore, since the expression inside the parentheses in Eq.~\ref{eq-def-cost} remains constant, the change in $f$ is given by $df = \sum_k (c_k - r \, Y_k(X)) \, d x_k$, which can be re-arranged to have $df = \sum_{k' \in U, k' \neq k_1} \left\{ (c_{k'} - r \, Y_{k'}(X)) - (c_{k_1} - r \, Y_{k_1}(X)) \right\} \, d x_{k'}$. The very definition of $k_1$ means all expressions in parentheses are non-negative, and since all $d x_k$'s are negative, we must have $df \leq 0$. Therefore, $\xx$ cannot be a local minimum of $f$. As such, $f$ will have - at most - $K$ local, single-ingredient minimums, one or more of which will be the global minimum.
\end{proof}

To put it in less mathematical terms, any multi-ingredient solution can be improved upon by replacing all sub-optimal ingredients with the optimal one while keeping the HRT - or equivalently $X$ - constant. Since the optimal ingredient is only a function of HRT, we can continue this substitution until we have made a full switch to the optimal biomass. The crucial requirement is that the ranking of biomasses does not change as we change $\xx$ while keeping $X$ constant. A sufficient condition for this to happen is to have constant marginal cost ($\cc$) for all biomasses, e.g., no volume discounts or supply limits.

In summary, thanks to Result~\ref{theorem-cost-minima}, the problem of minimising the $K$-dimensional cost function in Eq.~\ref{eq-def-cost} is reduced to minimising $K$, one-dimensional cost functions, which is computationally feasible. As such, the utility of Result~\ref{theorem-cost-minima} is that, by effectively providing the theoretical location of function minimums, it allows for accurate benchmarking of optimisation algorithms including QuAnCO.

%\textbf{Data Generation} \quad need more than real data to test an arbitrary problem size but still realistic, given knowledge of global and local minima this is a valuable benchmarking problem for the optimisation community, we want to disguise business secrets so we provide data generators with some modification from real data, see methods and supp-mat for details.

The biomass data - methane potentials and costs - used in experiments of Section~\ref{subsec-experiments} are simulated using multivariate distributions fitted to real data from NE. This setup allows us to produce realistic data - and hence test various optimisation algorithms - over a wide range of problem sizes, while - thanks to Result~\ref{theorem-cost-minima} - we always know the true minimum. See Section~\ref{sec-methods} for details.

%\subsection{Ising Solvers}\label{subsec-annealing}

%However, due to the physical limitations of the current-generation QAs (number and connectivity of qubits), the cost of conducting experiments for an iterative algorithm such as QuAnCO, and significant data transfer overhead via internet, we also implemented SA as an annealing-based proxy for QA, and used it for larger-scale feasibility experiments. Our implementation follows Ref~\cite{isakov2015optimised}, with a few improvements: using memory-aligned data structures better suited for dense matrices, batch random number generation, and multi-core parallelisation of independent samples drawn in each iteration.

%Embedding refers to mapping of a dense graph to a sparse one by representing some or all of the nodes in the dense graph via a `chain' of nodes in the sparse. Nodes within a chain must be connected using large enough weights to ensure they share a state, though making the connections too strong can consume the limited dynamic range of the hardware and reduce the quality of solutions.

\subsection{Experiments}\label{subsec-experiments}

D-Wave offers two QAs for public access~\cite{DWave2021Topology}: 1) the \emph{DW-2000Q} system (version 6) with 2,048 qubits connected in the `chimera' pattern, and 2) the \emph{Advantage} system (version 1.1) with 5,760 qubits, connected in the `pegasus' pattern. \textit{Advantage} has both higher qubit count and more connections per qubit than \textit{DW-2000Q}, thus allowing for embedding of larger and denser matrices with shorter chain lengths.

In order to test QuAnCO for larger problem sizes, we also implemented and used SA as an annealing-based proxy for QA. Our implementation follows Ref~\cite{isakov2015optimised}, with a few improvements: using memory-aligned data structures better suited for dense matrices, batch random number generation, and multi-core parallelisation of independent samples drawn in each iteration.

The key idea behind annealing in optimisation is to allow for state transitions to higher energy (or cost) in early iterations. Such transitions - which help with escaping shallow, local minima - are reduced in frequency as the algorithm progresses, such that it becomes a greedy, downhill search towards the end. In simulated annealing (SA), the source of noise is (simulated) thermal fluctuations~\cite{van1987simulated}, while in quantum annealing (QA), it is quantum tunneling ~\cite{morita2008mathematical}.

We also developed a high-performance `brute-force' or `exact' QUBO solver in C, which enumerates all permutations of the high-dimensional, binary decision space. This solver can be used in relatively small-size problems, e.g., $K \leq 25$. Using the exact Ising solver allows us to isolate the approximation effect of QA and SA solvers.

In the experiments below, we begin by comparing TRN against BFGS and CG. This is followed by comparing QuAnCO - using exact, QA and SA Ising solvers, in that order - against TRN. Experimental setup including data generation, tuning parameters for optimisation algorithms, and metric definitions can be found in Section~\ref{sec-methods}. Supp Mat includes results for quantum-inspired Ising solvers from Hitachi and Toshiba.

\textbf{Trust Region Newton} \quad Figure~\ref{fig-trn-vs-bfgs-main} compares the performance of TRN against BFGS and CG in terms of 1) solution quality, i.e., closeness of final results to true minimum (panels A-F), and 2) speed, i.e., total execution time to produce the final results (panel G). In terms of quality, we have TRN > CG > BFGS, while the ranking is reversed for speed. (All algorithms including BFGS perform better for smaller problems; see Supp Mat.) In other words, faster algorithms produce lower-quality solutions. In particular, it is encouraging to see that TRN is nearly risk-free compared to CG (and BFGS), i.e., it hardly performs significantly worse than CG, while often showing significant improvement (panels D-F). Also, Panel H of Figure~\ref{fig-trn-vs-bfgs-main} shows the percentage of time in TRN spent on solving the sub-problem. We see that the percentage is high for all three models, and increases with problem size, an indication of unfavourable scaling of this component vs. the Hessian calculation component.

\textit{Takeaways:} 1- Improving TRN speed could make it an attractive algorithm for unconstrained optimisation problems, 2- It is natural to give higher priority to improving the speed of the sub-problem solving step, as done in QuAnCO.

\begin{figure}[!h]
\centering
\includegraphics[width=\linewidth]{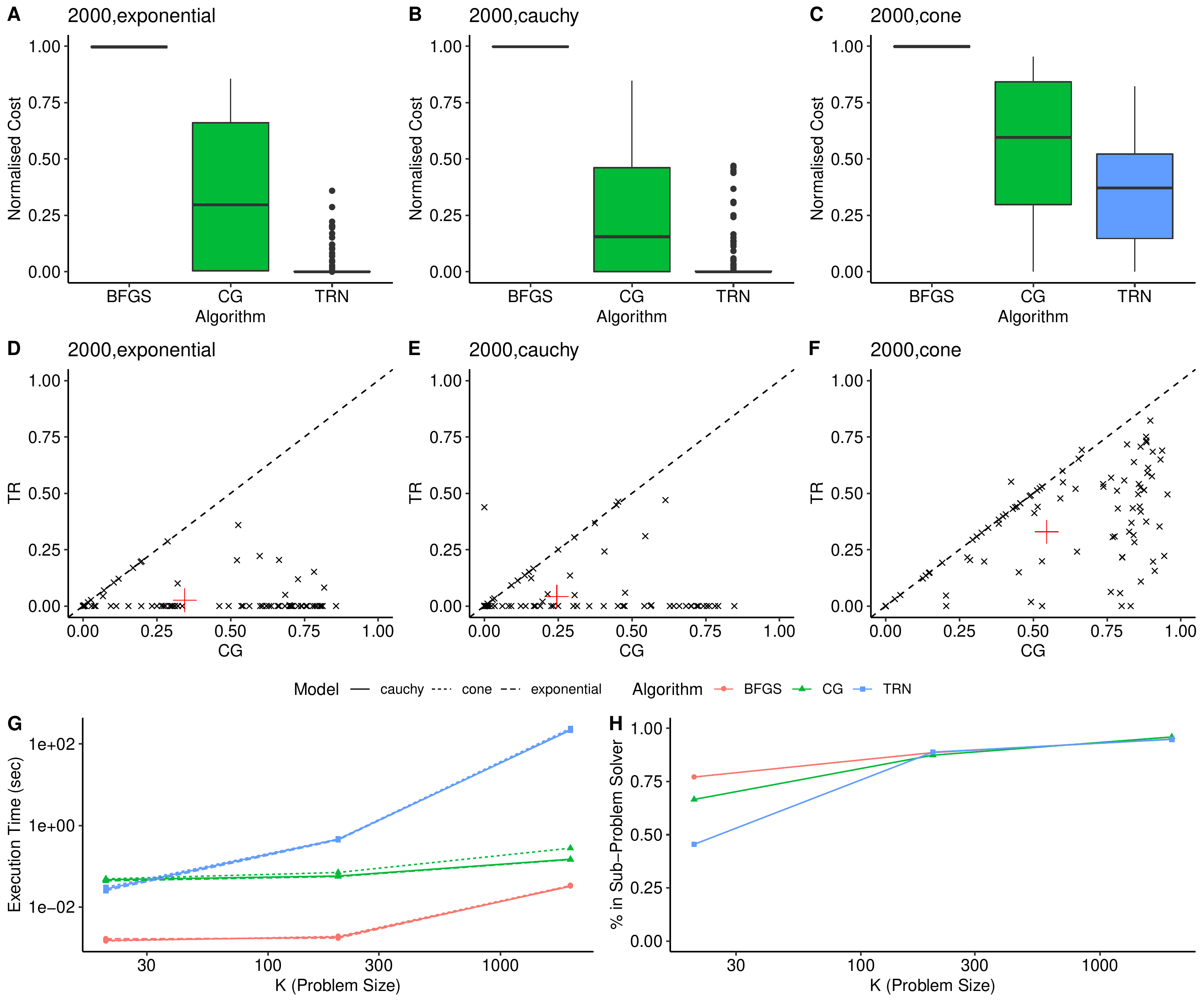}
\caption{Performance comparison of TRN vs. BFGS and CG. For all algorithms a maximum of 100 iterations were allowed (or convergence). \textbf{A-C}: Box plot of normalised costs for $K=2000$ and three BMP models. \textbf{D-F}: Normalised-cost cumulative probability curves, $K=2000$. \textbf{G}: Average execution time for the three algorithms and BMP models, with $K=20,200,2000$. \textbf{H}: Percentage of TRN time spent on solving the TR sub-problem for all three BMP models, $K=20,200,2000$.}
\label{fig-trn-vs-bfgs-main}
\end{figure}

\textbf{QuAnCO-Exact} \quad Figure~\ref{fig-trn-vs-qtro-brute} compares TRN against QuAnCO using an exact Ising solver (`QuAnCO-Exact'). As panel A shows, after the first few iterations, QuAnCO-Exact using $M=1,2,3$ performs better than TRN, though $M=1$ has a slower convergence than $M=2,3$. Using Table~\ref{table-qtro-exact}, we can see that, for example, at iteration 100 and $K=20$, QuAnCO-Exact produces a solution that is, on average, $\sim$9\% closer to the true minimum (measured from starting point) compared to TRN (p-value: $4\times10^{-5}$). This could be a very meaningful contribution to the bottom line of a renewable energy producer. As panel D shows, this improvement is nearly risk-free in this instance, i.e., the QuAnCO-Exact solutions are never meaningfully worse than TRN, but in many cases it offers a significant improvement over TRN.

\textit{Takeaways:} 1- For small problems where execution time is not a deciding factor, it may be worthwhile to use the QuAnCO-Exact algorithm instead of - or in conjunction with - TRN to help improve the quality of optimisation solutions, 2- We are encouraged by our observation that QuAnCO-Exact-1 produces nearly the same quality solutions as $M=2,3$. The ability to use smaller $M$'s is important as it leads to reduced time needed to compute $\Q$, less memory needed to store it, and a smaller data transfer overhead between CPU and the Ising-solver hardware. It also allows us to use QuAnCO-Exact for larger problems. (Recall that the $\Q$ matrix has dimensions $KM \times KM$.)

\begin{figure}[!h]
\centering
\includegraphics[width=\linewidth]{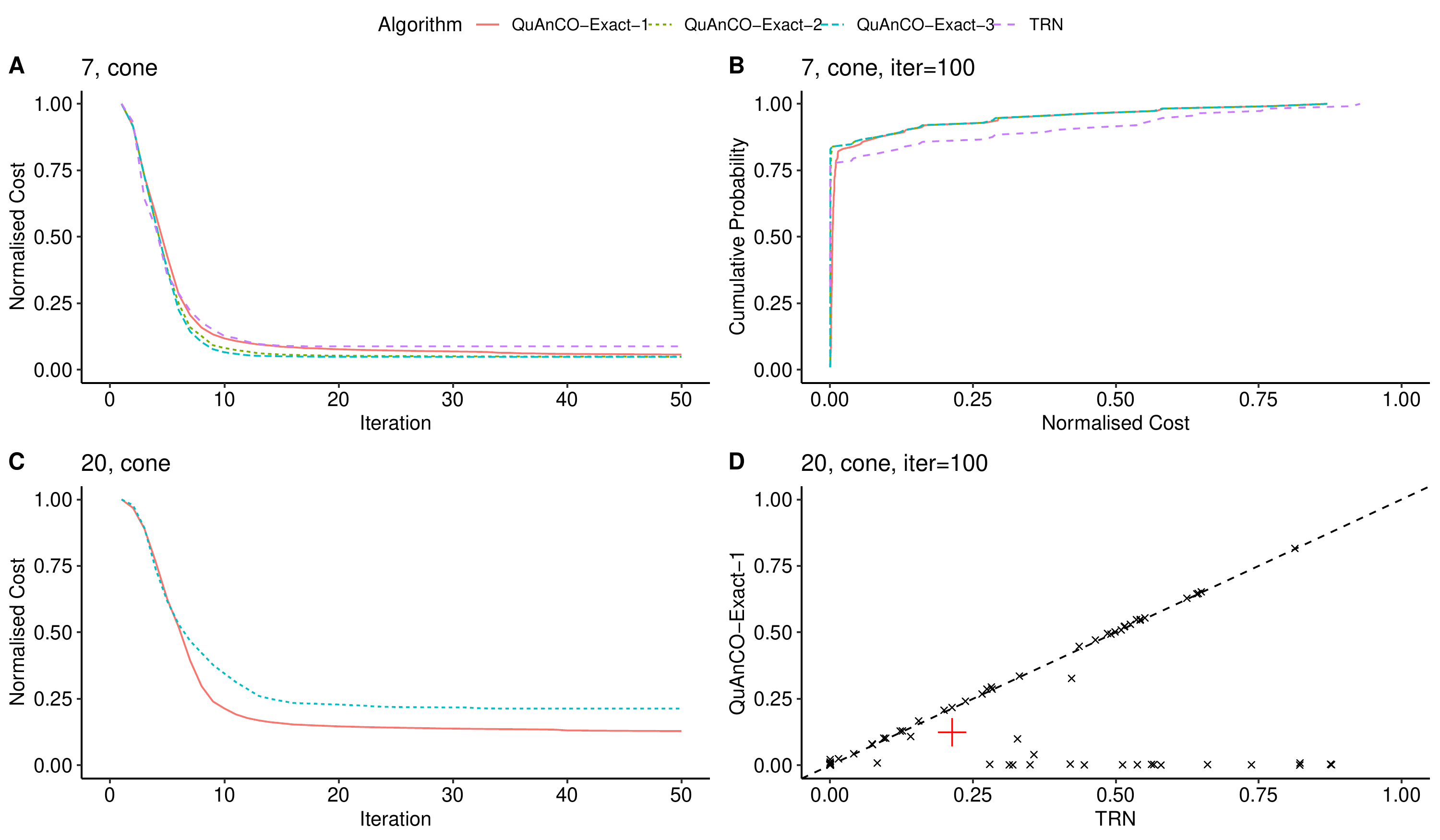}
\caption{Performance comparison of TRN vs. QuAnCO-Exact. Left panels show mean normalised cost vs. iteration number. Top Right: Cumulative probability curve for mean normalised cost for $K=7$ and three values of $M=1,2,3$. Bottom Right: Final mean normalised cost (iter 100) for QuAnCO-Exact-1 vs. TRN. The red `+' in panel D shows mean normalised cost at iteration 100.}
\label{fig-trn-vs-qtro-brute}
\end{figure}

\begin{table}[!h]
\centering
\begin{tabular}{rrrrrr}
\hline
    & & & \multicolumn{3}{c}{QuAnCO-Exact} \\
  \cline{4-6}
iter & K & TRN & M=1 & M=2 & M=3 \\ 
  \hline
10 &  3 & 11.3 & 11.7 & 8.5 & 7.4 \\ 
  10 &  5 & 11.2 & 11.2 & 8.0 & 6.2 \\ 
  10 &  7 & 12.7 & 11.7 & 8.2 & 6.5 \\ 
  10 & 20 & 34.5 & 21.4 &  &  \\ 
  100 &  3 & 10.1 & 7.0 & 6.7 & 6.7 \\ 
  100 &  5 & 10.3 & 5.5 & 5.1 & 5.0 \\ 
  100 &  7 & 8.8 & 5.3 & 4.8 & 4.8 \\ 
  100 & 20 & 21.3 & 12.5 &  &  \\ 
   \hline
\end{tabular}
\caption{Comparison of mean normalised cost for TRN vs. QuAnCO-Exact (using $M=1,2,3$ bits per dimension) after 10 and 100 iterations for $K=3,5,7,20$, using the cone model.}
\label{table-qtro-exact}
\end{table}

\textbf{QuAnCO-QA} Solving a QUBO on D-Wave's QAs requires choosing many tuning parameters, some of which are hardware specific. A few such parameters are listed and described in Supp Mat. In addition to the hardware controls, an important algorithmic parameter - induced by the need to embed a dense matrix in a sparse graph - is `chain strength', which sets the relative importance of the penalty term in the QUBO cost function that biases all physical qubits within each chain to have the same value (hence representing the same logical bit).

Figure~\ref{fig-default-perf} shows the effect of changing chain strength on the performance of a single-step QuAnCO using the \textit{Advantage} (left) and \textit{DW-2000Q} (right) devices. For both plots, a chain strength multiplier of 1.0 corresponds to leaving the chain strength value calculated using the D-Wave heuristic unchanged. We see that the optimal value for both devices is significantly smaller than 1.0. Based on this - albeit limited - experiment, we chose a value of 0.008 for the remaining experiments. Interestingly, at this `optimal' value of chain strength, we observe a significant rate of chain break, i.e., when different physical qubits representing the same logical bit do not have the same value, thus requiring some form of chain conflict resolution such as majority voting. We also note that the performance of the newer, \textit{Advantage} device is more sensitive to the chain strength multiplier, and on average does not match the older, \textit{DW-2000Q} device.

\begin{figure}[!h]
\centering
\includegraphics[width=500pt]{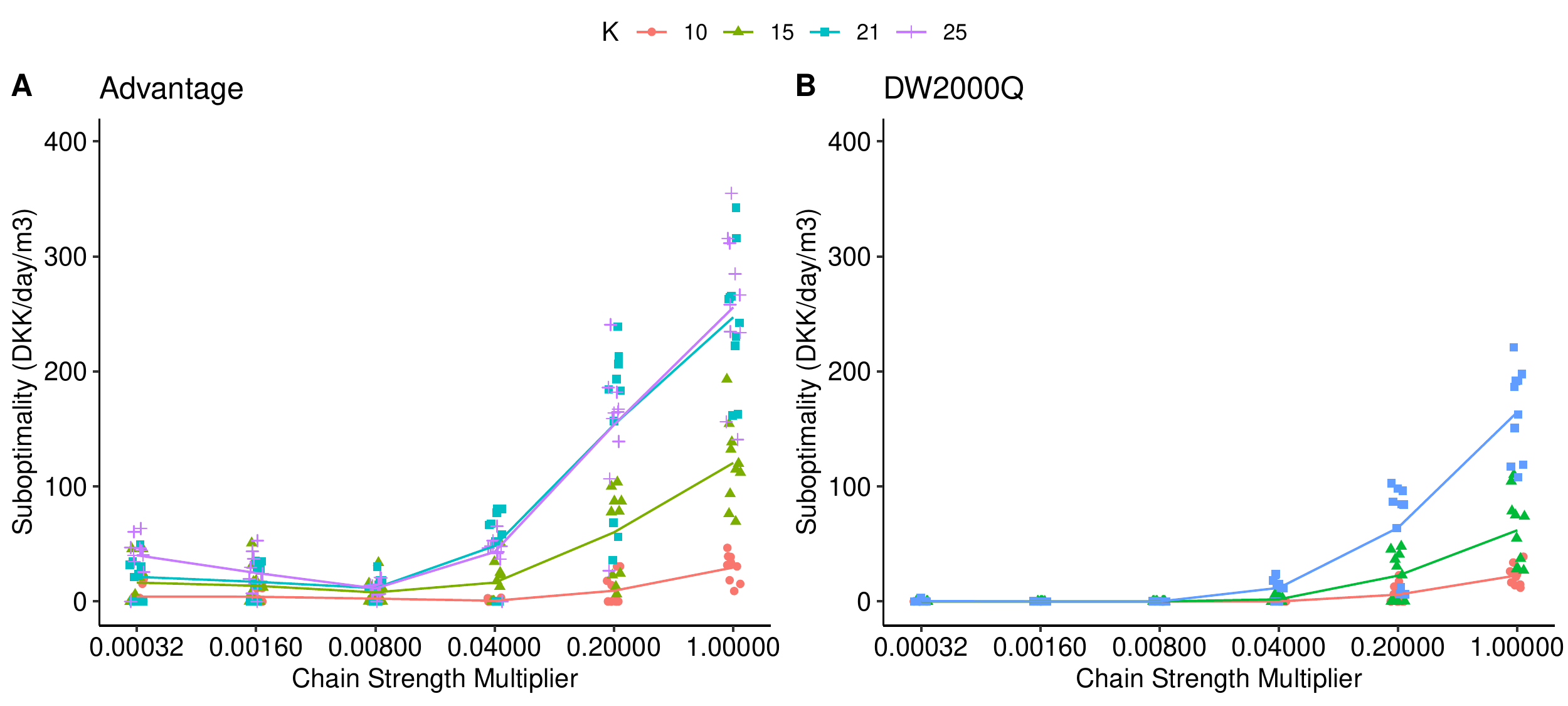}
\caption{Performance of \textit{Advantage} (left) and \textit{DW-2000Q} (right) QPUs for a range of problem sizes ($K=10,15,21,25$) and chain strength multipliers (sequence of $5^u$ with $u$ being an integer between $-5$ and $0$), using a single iteration of the QuAnCO algorithm. We used the `scaled' function to calculate the baseline chain strength parameter. For a given problem size, ten problems were used across all values of chain strength multiplier. Lines represent average of 10 runs for each problem size. We used $M=3$ bits per dimension for discretisation; therefore, problem sizes tested are 30, 45, 63, and 75 bits, with the last value tested for \textit{Advantage} only due to size limitation of \textit{DW-2000Q}. All other tuning parameters are left to their default values from D-Wave, and can be found in Table~\ref{table-dwave-parameters}.}
\label{fig-default-perf}
\end{figure}

Figure~\ref{fig-dwave-vs-trn-vs-sa} shows a comparison of TRN and QuAnCO-QA, using the \textit{minorminer} and \textit{clique} embedding libraries from D-Wave, applied to \textit{Advantage} and \textit{DW-2000Q} QPUs. All four combinations of QPU and embedding library show reasonable performance (and without much tuning). To our knowledge, these experiments provide the first evidence for feasibility of using quantum annealing in solving continuous optimisation problems.

\begin{figure}[!h]
\centering
\includegraphics[width=\linewidth]{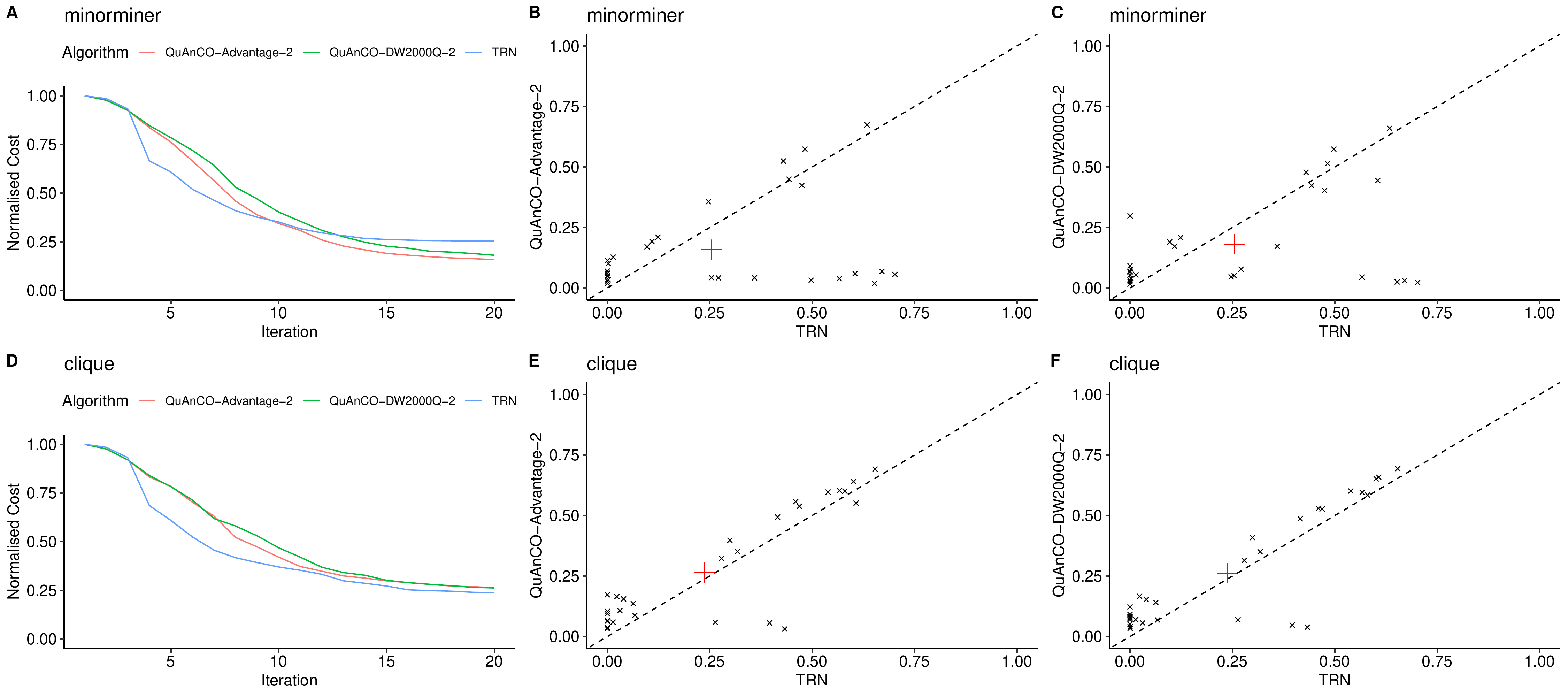}
\caption{Comparison of solution quality between TRN and QuAnCO-QA (using \textit{Advantage} and \textit{DW2000} devices, $M=2$ bits per dimension, and $K=30$. Top Row: Using \textit{minorminer} embedding library. Bottom Row: Using \textit{clique} embedding library. Left Column: Mean normalised costs (over 30 runs) for first 20 iterations. Middle and Right Columns: Final mean normalised costs (iteration 20) for QuAnCO-QA vs. QuAnCO-SA using both QPUs from D-Wave.}
\label{fig-dwave-vs-trn-vs-sa}
\end{figure}

We note that the performance gap between the two QPUs that we observed in Figure~\ref{fig-default-perf} (single iteration) is not evident in Figure~\ref{fig-dwave-vs-trn-vs-sa} for the multi-iteration experiment. A simple explanation is that our single-iteration experiments were conducted using version 1.1 of the \textit{Advantage} device, while the multi-iteration experiments used the updated version 4.1 of the QPU. It is also likely that the experiments conducted near the global minimum are more sensitive to noise and limited dynamic range of the QPUs, compared to experiments that start far from the global minimum. %Further research is needed to establish the factors that affect performance in local vs. global optimisation, and in the vicinity or far from a local or global minimum.

Aside from chain strength, adjusting other tuning parameters of these devices could improve their performance. Figure~\ref{fig-advantage-sensitivity} shows an example of the impact of these parameters on performance of the \textit{Advantage} device, using a single iteration of QuAnCO. (These experiments were conducted using the v1.1 of the device.) Note that `number of samples' and 'annealing time' both show a significant impact on performance. The impact of the first parameter is self-explanatory. Slowing down the annealing process by increasing annealing time helps the process better approximate the ideal, adiabatic process~\cite{morita2008mathematical}.

\begin{figure}[!h]
\centering
\includegraphics[width=\linewidth]{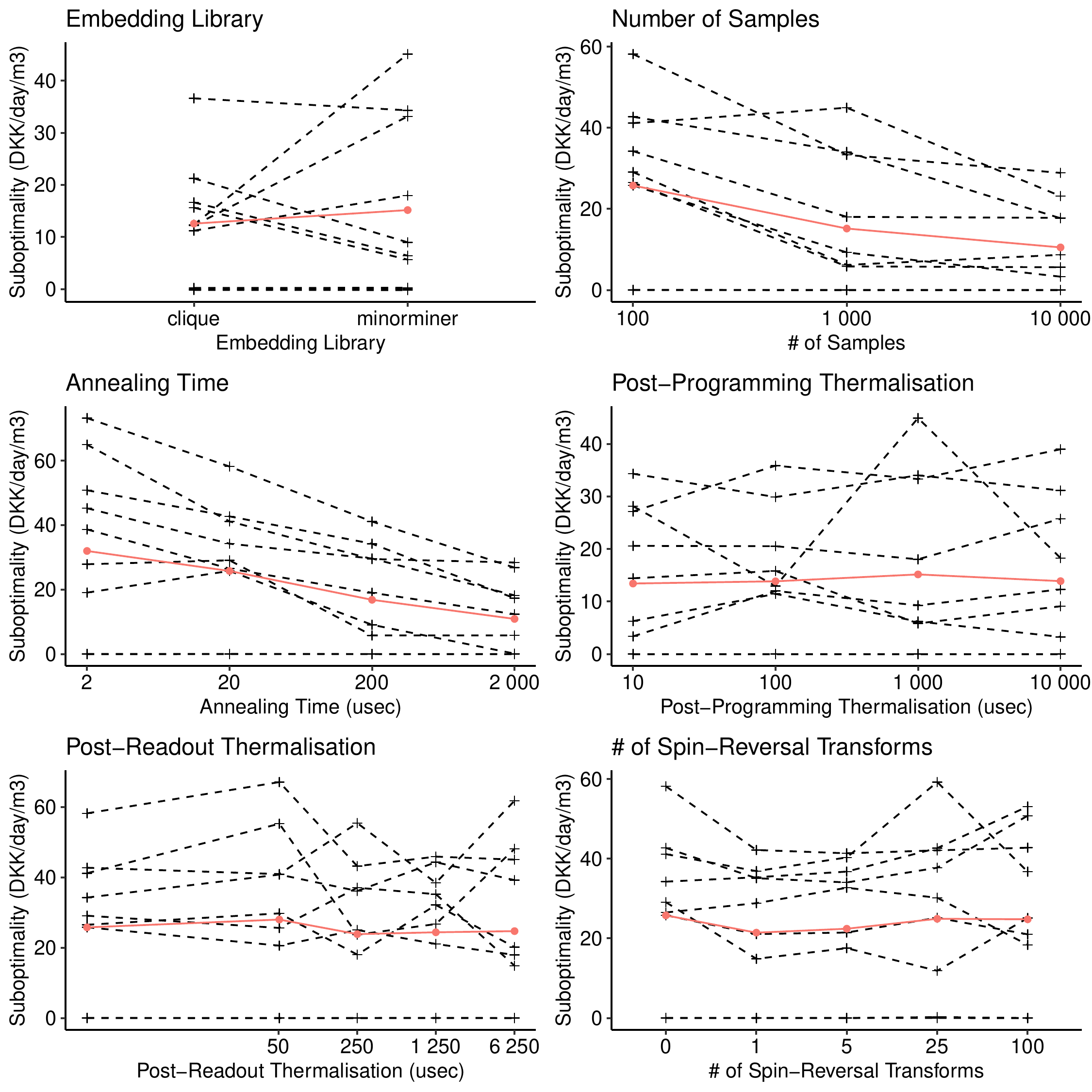}
\caption{Effect of various control parameters of D-Wave's \textit{Advantage} QPU on its performance for a single iteration of QuAnCO. Same ten problems were used in all figures. Number of biomasses ($K$) is 21, and number of digits ($M$) is 3, i.e., 63 logical bits. Chain strength multiplier was set to 0.008, based on results shown in Figure~\ref{fig-default-perf}.}
\label{fig-advantage-sensitivity}
\end{figure}

There are several types of overhead involved in using D-Wave's remote QAs for QuAnCO, including data transfer over the internet, queuing of QA jobs and device preparation. In current settings, the first two components can be multiple times larger than the core computation time. For example, while the so-called `QPU access time' for the Advantage device is less than 30msec (when drawing 100 samples per QA job), yet the end-to-end time to execute a single iteration of QuAnCO could reach as high as 30sec, i.e., a 1000x overhead.

In terms of how QPU access time scales with problem size, any extrapolation from current problem sizes ($K \leq 100$) would be highly speculative, especially since there will likely be tradeoffs between time and quality, as is the case with `annealing time' and `number of samples', shown in Figure~\ref{fig-advantage-sensitivity} to improve solution quality.

\textit{Takeaways:} 1- Despite limited tuning, QuAnCO-QA shows performance that is comparable to TRN in terms of solution quality, 2- The current setup of running QA via exchanging data with a shared resource over the internet imposes significant overhead, and prevents any speedup for an iterative algorithm such as QuAnCO, 3- More experiments - using larger and more densely-connected QAs - are needed to determine the scaling of the core annealing time for large problem sizes.

\textbf{QuAnCO-SA} \quad In order to test the feasiblity of QuAnCO for larger problem sizes, we use SA as an annealing-based proxy for QA. Unlike the exact solver, which scales exponentially with problem size, SA scales quadratically. Figure~\ref{fig-trn-vs-qtro-sa} compares QuAnCO-SA against TRN for $K=20,200,2000$, up to 50 iterations. For binarisation, $M=1,2$ bits per dimension were used. For the most difficult problem of the three (cone model), QuAnCO significantly outperforms TRN. For the other two models, TRN slightly outperforms QuAnCO.

\begin{figure}[!h]
\centering
\includegraphics[width=\linewidth]{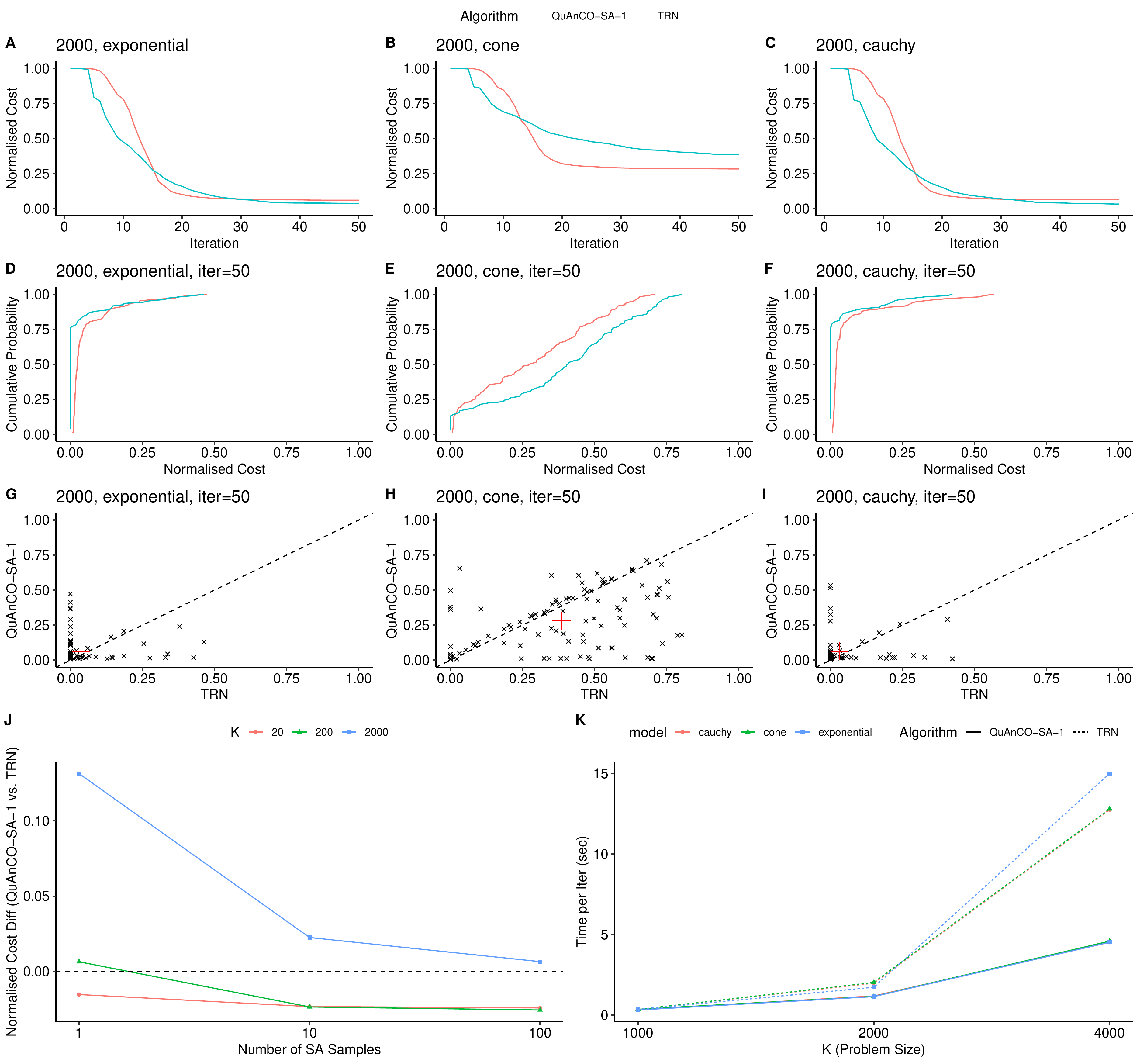}
\caption{Performance comparison of QuAnCO-SA-1 and TRN. Panels A-C: Normalised cost for first 50 iterations for three biomethane yield models and $K=2000$. Panels D-F: Cumulative probability plots vs. normalised cost, calculated at iteration 50. Panels G-I: Scatter plots of final normalised cost (at iteration 50) for QuAnCO-SA-1 vs. TRN for all three models. Panel J: Final mean normalised cost (iteration 50) for QuAnCO-SA-1 minus TRN, using 1,10,100 independent samples oer iteration for the SA-based Ising solver in QuAnCO. Panel K: Average time per iteration - calculated over the first 10 iterations - for QuAnCO-SA-1 vs. TRN for $K=1000,2000,4000$, and all three models.}
\label{fig-trn-vs-qtro-sa}
\end{figure}

An important setting in SA is the number of independent samples drawn per iteration (similar to QA), which was set to 10 in Figure~\ref{fig-trn-vs-qtro-sa} (panels A-I). The effect of changing this parameter is shown in Figure~\ref{fig-trn-vs-qtro-sa}, panel J, for the exponential model. We see that as problem size is increased, the effect of sample size becomes more prominent. Also note that SA samples can be drawn in parallel, e.g. via multi-threading on multicore CPUs, or using Graphic Processing Units (GPUs). This would allow us to improve the performance of QuAnCO-SA while keeping wall time nearly constant.

Panel K of Figure~\ref{fig-trn-vs-qtro-sa} compares time per iteration for QuAnCO-SA-1 and TRN. We have used 10 cores to parallelise the SA step. We see that scaling of TRN with problem size is significantly worse than QuAnCO-SA. This is because sub-problem solving in this TRN implementation~\cite{geyer2020trust} is dominated by eigen-decomposition of the Hessian, which scales like $O(K^3)$. The actual scaling is better - close to $O(K^{2.5})$ - since this computation is only needed after accepted proposals. On the other hand, QuAnCO-SA is dominated by SA, which scales like $O(K^2)$.

Taking together all the results shown in Figures~\ref{fig-trn-vs-qtro-sa}, we see that, for our renewable-energy optimisation problem, our current implementation of QuAnCO-SA outperforms TRN both in solution quality for a given iteration count, and also in time per iteration for large problems. Since can use only 1 bit per dimension, there is no overhead induced for calculating and movement of $\Q$. Also, note that we used 10 samples in SA, partly due to limitations of available cores on our test machine. Given that SA is nearly perfectly parallelisable, on systems with higher number of available cores, one can keep the execution time nearly constant while improving solution quality using a higher sample and core count.

\textit{Takeaways:} 1- QuAnCO-SA offers evidence of performance (combined speed and quality) advantage in using QuAnCO for large-scale, continuous optimisation problems, especially in difficult problems where TRN performs poorly, 2- Using massively-parallel hardware such as GPU, and HPC techniques, QuAnCO-SA can be further improved to offer a viable, short-term bridge to QA-based solutions that are larger, have denser connectivity, and are more accessible. 

\section{Conclusion}\label{sec-discussion}

Motivated by the current gap between QA research that is focused on discrete optimisation, and the continuous nature of many important optimisation problems in renewable energy, we have proposed the QuAnCO algorithm for optimisation of complex, non-convex cost functions in continuous spaces, using QA. Experiments using the biomass selection problem from Nature Energy confirmed the algorithmic feasibility of QuAnCO. The clear advantage of QuAnCO over TRN (and hence CG and BFGS methods) for the most difficult of the three cost functions studied (using the cone model) suggests that it may offer similar advantage when applied to even more realistic cost functions, e.g., those based on the ADM1 and other white-box models of AD~\cite{batstone2002iwa}.

Further research into QuAnCO using QA requires hardware advances (more qubits and/or higher connectivity) and better access modes (dedicated devices for research, co-location of CPU and QPU). In the near term, classical Ising solvers such as SA can offer performance advantage for large problems, especially after further enhancements such as GPU parallelisation, or single-instruction, multiple-data (SIMD) parallelisation on the CPUs\cite{mahani2015simd}. In addition, Quantum-inspired solvers such as Toshiba's SBM~\cite{goto2019combinatorial} may offer another path to the near-term application of QuAnCO. Near-term alternatives such as SA and SBM can also be used to conduct further research into QuAnCO, using other, real-world continuous optimisation problems from renewable energy and beyond.

Besides the strategy of applying element-wise nonlinearities to eliminate bound constraints (Section~\ref{subsec-trn-bound}), other solutions have been proposed and researched for handling bound constraints - as well as linear and nonlinear equality and inequality constraints - in TR algorithms. This includes a combination of penalty functions, Lagrangian methods, active-set strategies, sequential quadratic programming and relaxed linearisation~\cite{vardi1985trust,byrd1987trust,omojokun1989trust}. In Supp Mat, we outline how a combination of quadratic penalty terms and element-wise nonlinear transformation introduced in Section~\ref{subsec-trn-bound} can be used towards this end.

Another useful extension would be to include non-differentiable cost functions, which may arise in simulation-based systems. A natural avenue to explore is fitting quadratic response surfaces to the data~\cite{hao2018model}, which would automatically handle the first step in QuAnCO's sub-problem solver (see Supp Mat). An important decision in such an approach would be choosing the location and number of function evaluations to use for fitting the response surface.

In some applications, a single iteration of QuAnCO could be sufficient, e.g., when any adjustments to current operating conditions of an industrial or business process must be small and incremental due to risk-averseness. This could be, for instance, due to uncertainties in extrapolating the various black-box models that feed into the cost function beyond their historical range of parameters. In other words, the permissible search radius around the current point in the parameter space may be small enough to justify a quadratic approximation.

We believe that the QuAnCO algorithm presented in this paper takes a significant step towards utilising alternative computing platforms and hardware such as QA for solving important optimisation problems, including those related to green energy production. We envisage a rapid expansion and application of this framework, supported by improvements in various hardware technologies and their broader accessibility.

\section{Methods}\label{sec-methods}

\subsection{Quantum Annealing Continuous Optimisation (QuAnCO)}\label{subsec-methods-qnlp}

\textbf{Sub-Problem QUBO} \quad Combining Equations \ref{eq-x-n} and \ref{eq-n-Z}, we get:
\begin{equation}
    \xx = \aaa + \Delta \, \ZZ \, \bb
\end{equation}
Using the `vec' trick, $\mathrm{vec}(A \, B \, C) = (C^\top \otimes A) \, \mathrm{vec}(B)$, we can express $\xx$ as a linear function of $\zz = \mathrm{vec}(\ZZ)$, where $\mathrm{vec}(\ZZ)$ of the $K$-by-$M$ matrix $\ZZ$ is a vector of length $K \times M$ resulting from stacking the $M$ columns of $\ZZ$. Taking note that $\Delta \, \ZZ \, \bb$ is a vector of length $K$, and thus $\Delta \, \ZZ \, \bb = \mathrm{vec}(\Delta \, \ZZ \, \bb)$, we have:
\begin{equation}
    \Delta \, \ZZ \, \bb = (\bb^\top \otimes \Delta) \, \zz
\end{equation}
which is reflected in Equation~\ref{eq-def-x-x0}.

Deriving Eq.~\ref{eq-def-Q} from Eq.~\ref{eq-quad-approx} is mostly routine algebra, while noting that for a binary variable $z$, we have $z^2 = z$. Therefore, for a binary vector $\zz$:
\begin{equation}
    \mathbf{u}^\top \zz = \zz^\top \, \mathrm{diag}(\mathbf{u}) \, \zz
\end{equation}

\textbf{Denseness of $\Q$} \quad We prove that, under general conditions, the $\Q$ matrix is fully-dense (Result~\ref{theorem-denseness}).

\begin{proof}
Since the second term on the right-hand side of Eq.~\ref{eq-def-Q} is diagonal, we focus on the first term to prove denseness:
\begin{equation}
    \Am^\top \, \HH_0 \, \Am = \left( \bb^\top \otimes \Delta \right)^\top \, \HH_0 \, \left( \bb^\top \otimes \Delta \right)
\end{equation}
Note that
\begin{equation}
    \bb^\top \otimes \Delta = \begin{bmatrix}
      b_1 \, \Delta & b_2 \, \Delta & \dots & b_M \, \Delta
    \end{bmatrix}
\end{equation}
and therefore
\begin{equation}
    \HH_0 \, \left( \bb^\top \otimes \Delta \right) = \begin{bmatrix}
      b_1 \, \HH_0 \, \Delta & b_2 \, \HH_0 \, \Delta & \dots & b_M \, \HH_0 \, \Delta
    \end{bmatrix} %= \bb^\top \, \HH_0 \, \Delta
\end{equation}
Next, we take advantage of the distributive property of transposition over Kronecker product:
\begin{equation}
    \left( \bb^\top \otimes \Delta \right)^\top = \bb \otimes \Delta^\top = \begin{bmatrix}
      b_1 \, \Delta \\ b_2 \, \Delta \\ \vdots \\ b_M \, \Delta
    \end{bmatrix}
\end{equation}
Combining the above two, we obtain
\begin{equation}\label{eq-aha}
    \Am^\top \, \HH_0 \, \Am = \begin{bmatrix}
      b_1 \, \Delta \\ b_2 \, \Delta \\ \vdots \\ b_M \, \Delta
    \end{bmatrix} \, \begin{bmatrix}
      b_1 \, \HH_0 \, \Delta & b_2 \, \HH_0 \, \Delta & \dots & b_M \, \HH_0 \, \Delta
    \end{bmatrix} = \begin{bmatrix}
      b_1^2 \,\, \Delta \, \HH_0 \, \Delta & b_1 b_2 \,\, \Delta \, \HH_0 \, \Delta & \dots & b_1 b_M \,\, \Delta \, \HH_0 \, \Delta \\
      b_1 b_2 \,\, \Delta \, \HH_0 \, \Delta & b_2^2 \,\, \Delta \, \HH_0 \, \Delta & \dots & b_2 b_M \,\, \Delta \, \HH_0 \, \Delta \\
      \dots & \dots & \ddots & \dots \\
      b_1 b_M \,\, \Delta \, \HH_0 \, \Delta & b_2 b_M \,\, \Delta \, \HH_0 \, \Delta & \dots & b_M^2 \,\, \Delta \, \HH_0 \, \Delta
    \end{bmatrix}
\end{equation}
Finally, note that
\begin{equation}
    (\Delta \, \HH_0 \, \Delta)_{i,j} = \sum_k \sum_{k'} \Delta_{i,k} \, \HH_{0,k,k'} \, \Delta_{k',j} = \delta_i \, \delta_j \, \HH_{0,i,j}
\end{equation}
where we have taken advantage of $\Delta$ being diagonal. Since $\delta_i > 0, \forall i$ and also $b_k > 0, \forall k$, we conclude that $\Am^\top \, \HH_0 \, \Am$ is fully-dense as long as $\HH_0$ is fully dense.
\end{proof}
We also note that, Eq.~\ref{eq-aha} can be easily re-written as
\begin{equation}\label{eq-aHa-2}
    \Am^\top \, \HH_0 \, \Am = (\bb \, \bb^\top) \otimes (\Delta \, \HH_0 \, \Delta)
\end{equation}

\textbf{Discretisation Error} \quad Consider the one-dimensional case of Figure~\ref{fig-discretisation-error}. The worst-case scenario is for two adjacent nodes on the grid to both be the apparent minimums of a discretised, convex function, and for the absolute minimum - on the continuous scale - to lie halfway between the nodes. It is easy to see that the discretisation error, in this case, is $\frac{1}{8} \, \lambda \delta^2$, where $\lambda$ is the (positive) second derivative of the quadratic cost function and $\delta$ is grid resolution. The result below generalises this result to the multivariate case.

\begin{figure}[!h]
\centering
\includegraphics[width=10cm]{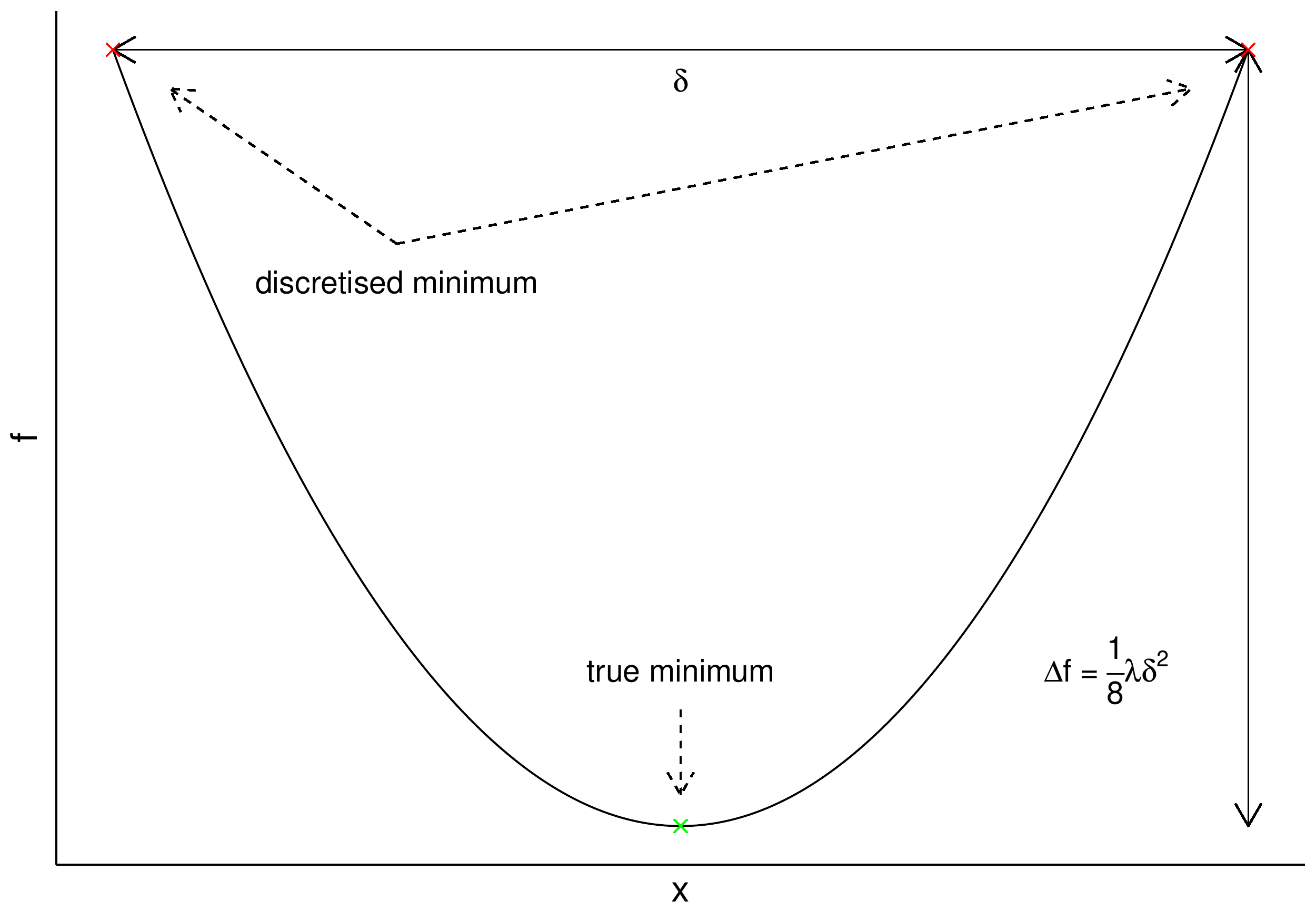}
\caption{Illustration of the upper bound on discretisation error - defined in Result~\ref{theorem-discretisation} - in the univariate case. $\lambda$ is the largest positive eigenvalue of the Hessian, which is also the second derivative of the function resulting from the intersection of the quadratic fit with a vertical plane. See text for details.}
\label{fig-discretisation-error}
\end{figure}

We calculate an upper bound on the error caused by replacing the continuous $\xx$ in Eq.~\ref{eq-quad-approx} with a discrete one per Eq.~\ref{eq-x-n}. The error would be due to the true minimum happening in-between grid points, and at a value lower than the observed minimum.

The worst-case scenario - i.e., maximum error between observed and true minimum - is when we have two neighboring points on the grid both being the observed minimums, and the true minimum lying halfway between them. Since we are minimising a quadratic function, its intersection with any plane is also of quadratic form. In particular, consider that $\xx$ is allowed to move along the direction indicated by the unit vector $\hat{\bbeta}$, around the origin located at $\aalpha$:
\begin{equation}
    \xx(u) = \aalpha + u \, \hat{\bbeta}
\end{equation}
Plugging the above back into Eq.~\ref{eq-quad-approx}, we obtain a quadratic form with the second derivative given by
\begin{equation}\label{eq-quad-proj}
    f''(u) = \sum_k \lambda_k \, \gamma_k^2
\end{equation}
where $\lambda_k$'s are the eigenvalues of $\HH$ and $\gamma_k$'s are the coefficients of $\hat{\bbeta}$ in the eigenbasis of $\HH$, i.e., $\hat{\bbeta} = \sum_k \gamma_k \, \hat{\vv}_k$ with $\hat{\vv}_k$ being the $k$'th eigenvector of $\HH$. To see the above, we simply note that
\begin{equation}
    \xx^\top \, \HH \, \xx = (\aalpha^\top + u \, \hat{\bbeta}^\top) \, \HH \, (\alpha + u \, \hat{\bbeta}) = \left( \hat{\bbeta}^\top \, \HH \, \hat{\bbeta} \right) \, u^2 + \hdots
\end{equation}
Using the eigendecomposition property, $\HH \, \hat{\vv}_k = \lambda_k \, \hat{\vv}_k$ readily leads to Eq.~\ref{eq-quad-proj}.

Noting that, for the unit vector $\hat{\bbeta}$ to have a norm of 1, we must have $\sum_k \gamma_k^2 = 1$, it is easy to conclude that maximum second deriative occurs when $\hat{\bbeta}$ points along the direction of largest positive eigenvalue of $\HH$, which we simply call $\lambda$ for brevity.

On the other hand, for a twice-differentiable cost function, minimums happen where the gradient vector is zero. For the function to have an interior minimum, it cannot be lower by more than $\frac{1}{2} \, \lambda d^2$, where $d$ is the minimum distance of the point from any of the vertices of the grid, which has a resolution vector $\mathbf{\epsilon}$. This minimum distance has a maximum of $\frac{1}{2} \, (\mathbf{\epsilon}^\top \, \mathbf{\epsilon})^{\frac{1}{2}}$. (See lemma below.) Combining the above two proves our result.

\begin{lemma}
Consider a point inside a $K$-dimensional hypercube of unit length. Let's call the minimum distance of the point from all $K$ vertices of the hypercube $d$. Maximum of $d$ is $\frac{\sqrt{K}}{2}$, which occurs when the point is at the center of the hypercube's main diagonal.
\end{lemma}

\begin{proof}
Without loss of generality, assume that the hypercube vertices having coordinates consisting of $0$'s and $1$'s only, i.e., it lies in the first orthant, its sides are aligned with coordinates, and has one vertex at the origin. The minimum distance of an interior point from the vertices - i.e., the inner optimisation problem - can be cast as a QUBO with the cost function $(\xx - \yy)^\top (\xx - \yy)$, where $\xx$ is the binary vector representing hypercube vertices, and $\yy$ is the location of the interior point. The resulting diagonal $\Q$ matrix associated with this QUBO is $\II_K - 2 \, \diag(\yy)$, and the cost function has a minimum of $\sum_k \min(1 - 2 \, y_k, 0)$. Therefore, the outer optimisation (maximisation) problem has the objective function $\sum_k \left\{ y_k^2 + \min(1 - 2 \, y_k, 0) \right\}$, subject to constraints $0 \leq y_k \leq 1, \,\, \forall k=1,\dots,K$. This can be easily verified to have a maximum of $\frac{\sqrt{K}}{2}$, located at $\yy = \frac{1}{2} \, \onek$. Generalisation to an orthotope or hyperrectangle can be done via a simple rescaling of coordinates.
\end{proof}

%As seen in Figure~\ref{fig-delta} (panel \textbf{D}), the upper bound is rather conservative in our experiments, where observed improvement from increasing grid resolution ($M=2$ to $M=5$) never exceeded more than $1/5$ of the upper bound. (Of course, it is also possible that increasing $M$ beyond $5$ would have significantly increased the observed improvement, but that is unlikely since often such successive improvements follow a geometric series.) It may be possible to calculate tighter upper bounds, but that must be weighed against any potential computational burden of such calculations, e.g., in each iteration of a global optimisation algorithm using QNLP.

\textbf{Bound Constraints} \quad To prove Eq.~\ref{eq-remove-box-constraints-a}, we start with Eq.~\ref{eq-nlp-yspace} and apply the chain rule of derivatives:
\begin{equation}
    \frac{\partial F}{\partial y_m} = \sum_n \frac{\partial f}{\partial x_n} \, \frac{\partial x_n}{\partial y_m} = \sum_n \frac{\partial f}{\partial x_n} \, \eta_m'(y_m) \delta_{m,n} = \frac{\partial f}{\partial x_m} \, \eta_m'(y_m),
\end{equation}
where we have taken advantage of each nonlinear function, $\eta_m$, being a function of $y_m$ only. Consolidating the above term for all $m$'s into a vector produces Eq.~\ref{eq-remove-box-constraints-a}. Taking the derivative of the two sides in the above yields:
\begin{equation}
    \frac{\partial^2 F}{\partial y_n \partial y_m} = \frac{\partial}{\partial y_n} \left( \frac{\partial f}{\partial x_m} \, \eta_m'(y_m) \right) = \frac{\partial}{\partial y_n} \left( \frac{\partial f}{\partial x_m} \right) \, \eta_m'(y_m) + \frac{\partial f}{\partial x_m} \, \frac{\partial}{\partial y_n} \left( \eta_m'(y_m) \right)
\end{equation}
Noting that
\begin{equation}
    \frac{\partial}{\partial y_n} \left( \eta_m'(y_m) \right) = \eta_m''(y_m) \, \delta_{m,n}
\end{equation}
and
\begin{equation}
    \frac{\partial}{\partial y_n} \left( \frac{\partial f}{\partial x_m} \right) = \sum_k \frac{\partial}{\partial x_k} \left( \frac{\partial f}{\partial x_m} \right) \, \frac{\partial x_k}{\partial y_n} = \sum_k \frac{\partial}{\partial x_k} \left( \frac{\partial f}{\partial x_m} \right) \, \eta'_n(y_n) \, \delta_{k,n} = \frac{\partial^2 f}{\partial x_n \partial x_m} \, \eta'_n(y_n).
\end{equation}
Combining the last three equations produces:
\begin{equation}
    \frac{\partial^2 F}{\partial y_n \partial y_m} = \frac{\partial^2 f}{\partial x_n \partial x_m} \, \eta'_m(y_m) \, \eta'_n(y_n) + \frac{\partial f}{\partial x_m} \, \eta''_m(y_m) \, \delta_{m,n}.
\end{equation}
Noting that $\eta'_m(y_m) \, \eta'_n(y_n)$ is simply the $(m,n)$'th element of the matrix $\mathbf{\eta} \, \mathbf{\eta}^\top$ allows to arrive at Eq.~\ref{eq-remove-box-constraints-b}.

In this paper, we use the following specific element-wise nonlinear functions:
\begin{equation}
    x_k = \eta_k(y_k) = 
    \begin{cases}
        y_k , \quad & a_k = -\infty, \,\, b_k = +\infty \\
        a_k + e^{y_k} , \quad & a_k > -\infty, \,\, b_k = +\infty \\
        b_k - e^{y_k} , \quad & a_k = -\infty, \,\, b_k < +\infty \\
        a_k + \frac{b_k - a_k}{1 + e^{-y_k}}, \quad & a_k > -\infty, \,\, b_k < +\infty
    \end{cases}
\end{equation}

%\subsubsection{Nonlinear Constraints}

\subsection{Biomass Selection Optimisation}\label{subsec-methods-biomass-problem}

\textbf{Parametric Yield Functions} \quad In the cost function of Eq.~\ref{eq-def-cost}, we need to specify the yield functions $\YY(X) = \begin{bmatrix} Y_1(X) & \dots & Y_K(X) \end{bmatrix}^\top$. The functions $Y_k()$ and their first and second derivatives can be described as
\begin{subequations}
\begin{align}
Y(X) &= G_0 \, y(\frac{V}{X}; \mathbf{\theta}) \\
Y'(X) &= G_0 \, (-\frac{V}{X^2}) \, y'(\frac{V}{X}; \mathbf{\theta}) \\
Y''(X) &= G_0 \, (\frac{V}{X^2}) \, \left( (\frac{2}{X}) \, y'(\frac{V}{X}; \mathbf{\theta}) + (\frac{V}{X^2}) \, y''(\frac{V}{X}; \mathbf{\theta}) \right)
\end{align}
\end{subequations}
where $V$ is the active volume of the reactor, thus making $V/X$ the Hydraulic Retention Time (HRT) of the reactor. $G_0$ is the maximum methane produced from a unit (volume) of the biomass, $\theta$ is the vector of (known) parameters of the production model, and $y()$ is the normalised yield curve, approaching $1$ as HRT goes to infinity. We have dropped the subscript $k$ from $y$, $Y$, and $\theta$ in the above to avoid clutter.

As mentioned in the literature~\cite{pererva2020existing}, there have been many parametric models for the BMP experiments proposed over the years, with no clear choice that has been shown to provide the best fit in all cases. In this paper, we test three parametric forms for the yield function: cone, exponential, and Cauchy. To ensure optimisation results are sensible, we selected parametric models where $y(t=0) = 0$, i.e., biomethane is not produced instantaneously.

\textit{Cone}:
\begin{subequations}
\begin{align}
y(t; \, k, n) &= \left[ 1 + (k \, t)^{-n} \right]^{-1} \label{eq-def-cone} \\
y'(t; \, k, n) &= (n \, k) \, (k \, t)^{-(n+1)} \, \left[ 1 + (k \, t)^{-n} \right]^{-2} \\
y''(t; \, k, n) &= (n \, k^2) \, (k \, t)^{-(n+2)} \, \left[ 1 + (k \, t)^{-n} \right]^{-2} \, \left\{ (2 \, n) \, (k \, t)^{-n} \, \left[ 1 + (k \, t)^{-n} \right]^{-1} - (n + 1) \right\}
\end{align}
\end{subequations}
(Note that, in the above, $k$ is one of the parameters of the cone mode, and not the biomass index.)

\textit{Exponential}:
\begin{subequations}
\begin{align}
y(t; \, \tau) &= 1 - e^{-t/\tau} \label{eq-def-exponential} \\
y'(t; \, \tau) &= \frac{1}{\tau} \, e^{-t/\tau} \\
y''(t; \, \tau) &= -\frac{1}{\tau^2} \, e^{-t/\tau}
\end{align}
\end{subequations}

\textit{Cauchy}:
\begin{subequations}
\begin{align}
y(t; \, \tau) &= \frac{2}{\pi} \, \arctan{(\frac{t}{\tau})} \label{eq-def-cauchy} \\
y'(t; \, \tau) &= \frac{2}{\pi \, \tau} \, \frac{1}{1 + (t / \tau)^2} \\
y''(t; \, \tau) &= \frac{-4}{\pi \, \tau^2} \, \frac{t / \tau}{\left( 1 + (t / \tau)^2 \right)^2}
\end{align}
\end{subequations}

To calculate the gradient and Hessian of $f()$, we note that
\begin{subequations}
\begin{align}
& \partial x_{k'} / \partial x_k = \delta_{k,k'}, \\
& \partial X / \partial x_k = 1
\end{align}
\end{subequations}
Applying the chain rule, we get:
\begin{equation}
\frac{\partial f}{\partial x_k} = - c_k + r \, \left( Y_k(X) + \sum_{l=1}^K Y'_{l}(X) \, x_{l} \right)
\end{equation}
with $Y'_k(X) \equiv dY_k/dX$. In vector form:
\begin{equation}
\g_0 = - \cc + r \, \left( \YY + (\YY'^\top \xx) \, \one \right)
\end{equation}
More algebra leads to the following second derivative expression:
\begin{equation}
\frac{\partial^2 f}{\partial x_k \, \partial x_{k'}} = r \, \left\{ Y'_k(X) + Y'_{k'}(X) + \sum_{l=1}^K Y''_l(X) \, x_l \right\}
\end{equation}
In vector form:
\begin{equation}\label{eq-H-constituents}
\HH_0 = r \left\{ \YY' \, \one^\top + \one \, \YY'^\top + (\YY''^\top \xx) \, \mathbf{J} \right\},
\end{equation}
where $\mathbf{J}$ is a $K \times K$ matrix of ones: $\mathbf{J} \equiv \one \, \one^\top$. As expected, the Hessian is symmetric, i.e., $\frac{\partial^2 f}{\partial x_k \, \partial x_{k'}} = \frac{\partial^2 f}{\partial x_{k'} \, \partial x_{k}}$. It is easy to verify that $\HH_0$ is fully-dense since each of the three constituent terms in Eq.\ref{eq-H-constituents} are fully-dense, and there is no reason for their corresponding elements to cancel out in general.

\subsection{D-Wave Quantum Annealers}

For quantum annealing, we use two Quantum Processing Units (QPUs) from D-Wave - accessed via Amazon Braket on AWS: the \emph{DW-2000Q} system (version 6) with 2,048 qubits connected in the `chimera' pattern, and 2) the newer, \emph{Advantage} system with 5,760 qubits, connected in the `pegasus' pattern. \textit{Advantage} has both higher qubit count and more connections per qubit than \textit{DW-2000Q}, thus allowing for `embedding' of larger and denser matrices with shorter chain lengths. For the \textit{DW-2000Q} QPU, version 6 was used in all experiments. For \textit{Advantage} QPU, version 1.1 was used in the single-iteration experiments (local optimisation), while version 4.1 (replacing version 1.1 during our research) was used in the multi-iteration (global optimisation).

For embedding and unembedding, we use two libraries provided by D-Wave, \textit{minorminer} and \textit{clique}, offered as part of the \textit{ocean} Python SDK.

%\subsection{Simulated Annealing}\label{subsec-methods-qa}

\subsection{Experiments}

Experiments in this paper fall under two categories: single-iteration (SI) experiments (local optimisation), and multi-iteration (MI) experiments (global optimisation). The SI experiments were focused on D-Wave QPUs. (Conducting extensive MI experiments with D-Wave is rather costly, since a single run may require up to 100 individual QA jobs - one per iteration.) Therefore the SI experiments were limited in size to less than $30$ biomasses, and we were able to use real data for them. The MI experiments, on the other hand, included other Ising solvers and a much wider range of problem sizes (up to $K=2000$ biomasses). This necessitated using synthetic data that were made to resemble the real data in the statistical sense. Likewise, initialisation and other aspects have differences between the SI and MI experiments. We describe them below.

\textbf{Data Generation - SI} \quad Biochemical methane potential (BMP) tests were used to determine the methane potential and
biodegradability of the different biomass substrates in batch mode using the Automatic Methane Potential Test System II (AMPTS® II) from Bioprocess Control. The biomass substrates were obtained from various local distributors in Denmark. In the test, a substrate was mixed with a filtrated anaerobic bacteria culture (substrate-to-inoculum ratio 1:4 to 1:2) freshly retrieved from an active primary digester from a local biogas plant in Denmark (NE Midtfyn A/S). Total solids (TS) and volatile solids (VS) were measured on all substrates according to standard procedures. 500 ml bottles containing inoculum and the substrate were kept under thermophilic conditions at a temperature of 52 °C and mixed in a 60 sec on/60 sec off sequence for a period of 30 days. Methane and carbon dioxide are produced during the testing period due to the anaerobic degradation of the organic contents of the substrate. The methane generated from the substrate is then measured and methane production of the substrate is normalised by the mass of volatile solids added (Nml CH 4 /g VS), and after subtracting the methane production from a container of Microcrystalline cellulose, used as control.

The experimental data is fitted with the Cone, Exponential, and Cauchy models (Section~\ref{subsec-methods-biomass-problem}) by minimizing the sum of squared differences between the experimental and calculated values (nonlinear curve fitting).

The cost per fresh weight used in optimisation experiments is the sum of two components: procurement cost and transportation cost. Both components were calculated as average across sourcing and delivery sites, using numbers provided by NE's finance department for the harvest year 2020-2021. Unit revenue of biomethane is based on the 2020 figure provided by NE, and is the sum of market price of methane (1.5 DKK/Nm3) and subsidies (4.5 DKK/Nm3).

\textbf{Data Generation - MI} \quad At a high level, the data-generation process for MI experiments works by 1) fitting multivariate normal distributions (MVNs) to Nature Energy raw data (see above), on a transformed scale, 2) drawing samples from the fitted MVNs and, 3) transforming the samples back to the original scale. This is done separately for each of the three biomethane yield models, i.e., cone, exponential and Cauchy. We illustrate this for the cone model. Steps for exponential and Cauchy models are similar.

We define a `cost margin' parameter, $\alpha = c / (r \times G_0)$, where $c$ is the cost per tonne of fresh weight for the biomass, $r$ is the expected revenue per unit of biomethane produced, and $G_0$ is maximum expected biomethane released per fresh tonne of biomass (when HRT approaches infinity). We used $r=6.0\,DKK/Nm^3$ throughout the paper, where $DKK$ stands for Danish Krone. We apply the following transformation to NE data:
\begin{equation}\label{eq-datagen-fwd}
    (\alpha', G_0', n', k') = (\log(\frac{\alpha}{1 - \alpha}), \log(G_0), \log(n), \log(k))
\end{equation}
Next, we fit a 4-dimensional MVN to the NE data across all biomasses. After drawing samples - each sample being a synthetic biomass - from the fitted distribution, we reverse-transform the samples to match the forward-transform in Eq.~\ref{eq-datagen-fwd}. Since for some biomasses the cost is zero (thus causing the transformed value to become infinite), we remove them from the data used for fitting MVNs. Note that we need $K$ samples to support a simulation that uses $K$ biomasses. The samples, i.e., biomasses, are generated independently from one another.

Finally, in order to avoid extreme properties for the synthetic biomasses, we perform univariate optimisation on individual biomasses in each synthetic dataset, and reject those biomasses where the univariate optimal value falls outside the interval, $[0.01, 100]$.

\textbf{Initialisation - SI} \quad For panel D of Figure~\ref{fig-biomass-diversity}, number of biomass types was given values of 3, 9, 15, 21 (increments of 6). For each value, $1000$ runs was conducted. In each run, a random subset - including a random order - of biomass types from NE database were selected. The first value for $\xx_0$ was chosen to be the univariate optimal value, i.e., the value that would maximise the reactor profit (unit volume) assuming that only that single biomass type was fed into the reactor. For all other biomass types, their corresponding entry in $\xx_0$ was set to zero.

For Figure~\ref{fig-default-perf}, number of biomass types was given values of 10, 15, 21, 25. (The last value was only used with the \textit{Advantage} QPU, due to the physical limitations of the \textit{DW-2000Q} QPU.) For each combination of QPU, biomass type count and QA tuning parameters (chain strength multiplier and number of samples), we conducted 10 runs. In each run, we used the biomass with maximum univariate profitability as our `anchor' biomass (`Deep Bedding - Chicken'). Next, we randomly selected the remaining biomass types, designating the first of those as `contaminating' biomass and the rest as `nuisance' biomass. All nuisance biomass types were given a value of zero for their $\xx_0$ entry. For anchor biomass, its daily rate was selected from a uniform distribution centered on its univariate optimal value, with a range of $\pm 0.15$. For contaminating biomass, its daily rate was selected from a uniform distribution with a minimum of $0.0$ and a maximum of $0.15$. With this setup, as we increase problem size (number of biomass types), the true minimum of the cost function within our rectangular neighborhood is always located at the univariate optimal point of the anchor biomass. Initial point is suboptimal due to 1- random deviations of anchor biomass from optimal value, 2- non-zero value for contaminating biomass. As we increase problem size (number of nuisance biomass types), QA will be tested to `ignore' the nuisance dimensions and find the true minimum.

For Figure~\ref{fig-advantage-sensitivity}, same 10 problems corresponding to the previous figure and 21 biomasses were used.

\textbf{Initialisation - MI} \quad The starting point for all MI experiments was chosen such that 1) the resulting HRT was always 10 days, and 2) all biomasses would have equal daily volume. In other words, $\xx_0 = \frac{1}{10K} \onek$ (which causes $X_0 = \sum_k x_{0,k} = 1/10$). As a reminder, $\xx$ represents daily feed rate of biomasses \textit{per unit reactor volume}.

\textbf{Software and Hardware Environments} \quad Implementation of the core QuAnCO algorithm was done in R~\cite{R2022Core}, with exact and SA Ising solvers written in C++ due to performance requirements. Python scripts were used for interfacing with Amazon Braket (D-Wave quantum annealers) and Hitachi's CMOS annealing API. For TRN, the `trust' package in R was used~\cite{geyer2020trust}, while the `optim' function in the core R library was used for BFGS and CG methods.

\textit{D-Wave} \quad For the \textit{DW-2000Q} QPU, version 6 was used. For the \textit{Advantage} QPU, version 1.1 was used for SI experiments, while version 4.1 was used for MI experiments.

\textit{Toshiba} \quad We used version 1.2.3, available on AWS Marketplace, running on a p3.2xlarge instance.

\textit{Hitachi} \quad V2 of the annealing cloud web API was used. We tested device numbers 4 (GPU, 32-bit, float) and 5 (ASIC, 4-bit).

\textbf{Tuning Parameters - SI} \quad For generating Figure~\ref{fig-default-perf}, we used the default parameters listed in Table~\ref{table-dwave-parameters}, which also describes how the tuning parameters were changed to generate each of the six panels of Figure~\ref{fig-advantage-sensitivity}. The value of `chain strength multiplier', which was varied in Figure~\ref{fig-default-perf} to form the x-axis, was fixed at 0.008 for generating Figure~\ref{fig-advantage-sensitivity}.

\begin{table}[!h]
\scriptsize
\centering
\begin{tabular}{@{}llllllll@{}}
\toprule
Parameter                      & Range/Values & 1 & 2 & 3 & 4 & 5 & 6 \\ \midrule
Embedding library  & M, C & \textbf{M,C} & M & <- & <- & <- & <- \\
Number of reads             & 1-1e4 & 1e3 & \textbf{1e2,1e3,1e4} & 1e2 & 1e3 & 1e2 & <- \\
Annealing time ($\mu sec$)              & 1-2e3 & 2e1 & <- & \textbf{2e1, 2e2, 2e3} & 20 & <- & <- \\
Programming thermalisation ($\mu sec$)  & 0-1e4 & 1e3 & <- & <- & \textbf{1e1,1e2,1e3,1e4} & 1e3 & <- \\
Readout thermalisation ($\mu sec$) & 0-10,000 & 0 & <- & <- & <- & \textbf{0,50,250,1250,6250} & 0 \\
Spin reversal transforms & 0 - <number of reads> & 0 & <- & <- & <- & <- & \textbf{0,1,5,25,100} \\
\bottomrule
\end{tabular}
\caption{List of key control parameters for D-Wave QPUs. For definitions, see Supp Mat. Columns 1-6 indicate the values of each parameter used for each experiment in Figure~\ref{fig-advantage-sensitivity}. Values in bold correspond to the parameter whose value was varied in that experiment. The symbol `<-' means same value as prior experiment was used. Abbreviations: M = minorminer, C = clique. The value of chain strength multiplier parameter was fixed at 0.008 in all experiments reported in Figure~\ref{fig-advantage-sensitivity}, based on results reported in Figure~\ref{fig-default-perf}.}
\label{table-dwave-parameters}
\end{table}

\textbf{Tuning Parameters - MI} \quad For D-wave (Figure~\ref{fig-dwave-vs-trn-vs-sa}), the default values listed in Table~\ref{table-dwave-parameters} were used, with the following overrides: `chain strength multiplier' (not listed in the table) was set to 0.008, and `number of reads' was 100.

\textit{Simulated Annealing} \quad Initial inverse temperature: 0.1; final inverse temperature: 3.0; number of sweeps (from initial to final temperature): 100; number of independent samples: 10 (except for panel J of Figure~\ref{fig-trn-vs-qtro-sa} where it was given values of 1, 10, 100). For temperature trajectory over sweeps, we applied a log-linear regime, i.e., logarithm of inverse temperature was grown linearly from initial to final value.

\textit{Toshiba SBM} \quad Number of loops: 1; `prefer': auto, `dt': 1.0; C: auto selected.

\textit{Hitachi CMOS annealer} \quad Initial temperature: 10.0; final temperature: 0.01; number of steps: 10; step length: 100; chain strength multiplier: 0.008.

\textbf{Performance Metrics} \quad For single-iteration experiments, we define `suboptimality' of the results as the final value of the cost function (returned by QuAnCO) minus the true minimum. The units are the same as the cost function, which is $DKK/day/m^3$, reflecting the net daily cost (in Danish Krone) of biomethane production (biomass cost minus methane revenue), per unit reactor volume.

For multi-iteration experiments, we define `normalised cost' as the difference between the final cost and true minimum, divided by the difference between initial cost (at $\xx_0$) and the true minimum. Since in QuAnCO, as in all TR algorithms, we are guaranteed to improve or stay the same in every iteration, the normalised cost metric will always fall between 0.0 and 1.0. A normalised cost of 0.0 means the algorithm found the true minimum, while a value of 1.0 means no improvement was made compared to the initial value.

\section{Supplementary Material}\label{sec-supp-mat}

\subsection{Quantum Annealing}

QA is based on the adiabatic theorem which, in simplified terms, states that a quantum mechanical system starting in an eigenstate - e.g., ground state - of a slowly-changing Hamiltonian will remain in its corresponding eigenstate during system evolution. A transverse-field Ising-model implementation starts out the system - consisting of a collection of interacting quantum bits or qubits - in the initial/tunneling Hamiltonian with a trivial ground state, and slowly introduces the final/problem Hamiltonian, corresponding to the QUBO being minimised. By the end of annealing, only the problem Hamiltonian remains, at which point the qubit spins are `measured' and the classical bits corresponding to the minimum energy - or a value close to it - are returned:
\begin{equation}
    \mathcal{H}_{ising}(s) = - \frac{A(s)}{2} \left( \sum_i \hat{\sigma}_x^{(i)} \right) + \frac{B(s)}{2} \left( \sum_i h_i \, \hat{\sigma}_z^{(i)} + \sum_{i > j} J_{i,j} \, \hat{\sigma}_z^{(i)} \, \hat{\sigma}_z^{(j)} \right)
\end{equation}
In the above, $s$ is the normalised annealing time varying between 0 and 1. $A(s)$ and $B(s)$ control the relative strength of initial and final Hamiltonians, such as $A(0) >> B(0)$ and $A(1) << B(1)$. Quantum annealing is similar to (classical) simulated annealing (SA), with the key difference being that state transitions in QA are due to quantum tunneling rather than thermal fluctuations in SA.

Next, we present a brief overview of a few important tuning parameters in D-Wave's QAs.

\textbf{Number of reads} or samples determines how many anneal-readout cycles QPU should perform. Collecting more than one sample is important since various sources of error may cause a single anneal cycle to deviate from the ground state of the intended problem Hamiltonian. The best result among all samples is often chosen as the final output from QA solver. Time-series of sample energies can be examined for evidence of noise accumulation, e.g., due to insufficient heat dissipation after each readout.

\textbf{Embedding library} refers to the algorithm for mapping the problem graph - corresponding to $\Q$ - to the QPU connectivity graph, such that the problem graph becomes a minor of the QPU graph. The goal is to create indirect connections between vertices that cannot be directly connected due to sparseness of the QPU graph. Two libraries are provided by D-Wave: a heuristic, general-purpose library called minorminer embedding~\cite{cai2014practical}, and a specialised one for embedding fully-connected graphs, called clique embedding~\cite{boothby2016fast}. Note that embedding itself can be time consuming, especially for large graphs, but for fully-dense matrices it can be pre-calculated once for each graph size.

\textbf{Chain strength} controls the relative importance of chain-coherence penalty term vs. the original QUBO cost function. Low values lead to frequent chain breaks and hence inconsistent qubit values representing a single logical bit, while high values lead to mis-utilisation of limited dynamic range of QPU coupling strengths and insensitivity of final solutions to the desired cost function. D-Wave offers two heuristics for setting chain strength: 1) \textit{scaled}, which sets chain strength equal to maximum absolute value of $\Q$ entries, and 2) a so-called \textit{uniform torque compensation} method.

\textbf{Chain break resolution} determines the logic for resolving inconsistent bits within a chain. Options offered by D-wave are: 1) discard any samples with any broken chain, 2) fix each broken chain by taking majority vote among chain bits, 3) use weighed random choice, 4) minimize local energy. For dense, high-dimensional problems involving many chains, most samples are bound to have at least one chain break and hence the discard strategy is not practical in those cases.

\textbf{Annealing time} determines the speed with which the Hamiltonian evolves from the tunneling term to the problem term (see above). Previous research~\cite{jansen2007bounds,lidar2009adiabatic} indicates that smaller annealing times are needed to maintain performance for problems with small minimum spectral gap. However, increased annealing time has an obvious computational cost.

\textbf{Programming thermalisation} defines the length of delay (in microseconds) after `programming' the QPU with the $\Q$ matrix values, and before starting the annealing cycles. This delay allows for dissipation from generated in the programming process. This is done only once at the beginning of the anneal-readout cycles. For both \textit{DW-2000Q} and \textit{Advantage} devices, the acceptable range of this parameter is 0-10,000 microseconds, with a default value of 1,000 microseconds.

\textbf{Readout thermalisation} defines the length of time (in microseconds) at the end of each anneal-readout cycle to pause before starting the next anneal. This allows for heat removal caused by reading out the qubit values. For both \textit{DW-2000Q} and \textit{Advantage} devices, the acceptable range for this parameter is 0-10,000 microseconds, with a default value of zero.

\textbf{Post-processing} refers to strategies for improving quality of solutions returned by QPU, e.g., by doing a local search in the neighborhood of each sample returned by the QPU, using classical solvers.

\textbf{Spin reversal transforms} helps reduce the effect of programming biases and errors, but requires programming the QPU for each sample, hence adding to total sampling time and amount of heat generated in the device.

\subsection{Sub-Problem Solver in QuAnCO}
Figure~\ref{fig-qnlp-overview} illustrates the steps involved in solving the TR sub-problem of the QuAnCO algorithm in details. The full QuAnCO is listed in Algorithm~\ref{alg-qtro}.

\begin{figure}[!h]
\centering
\includegraphics[width=\linewidth]{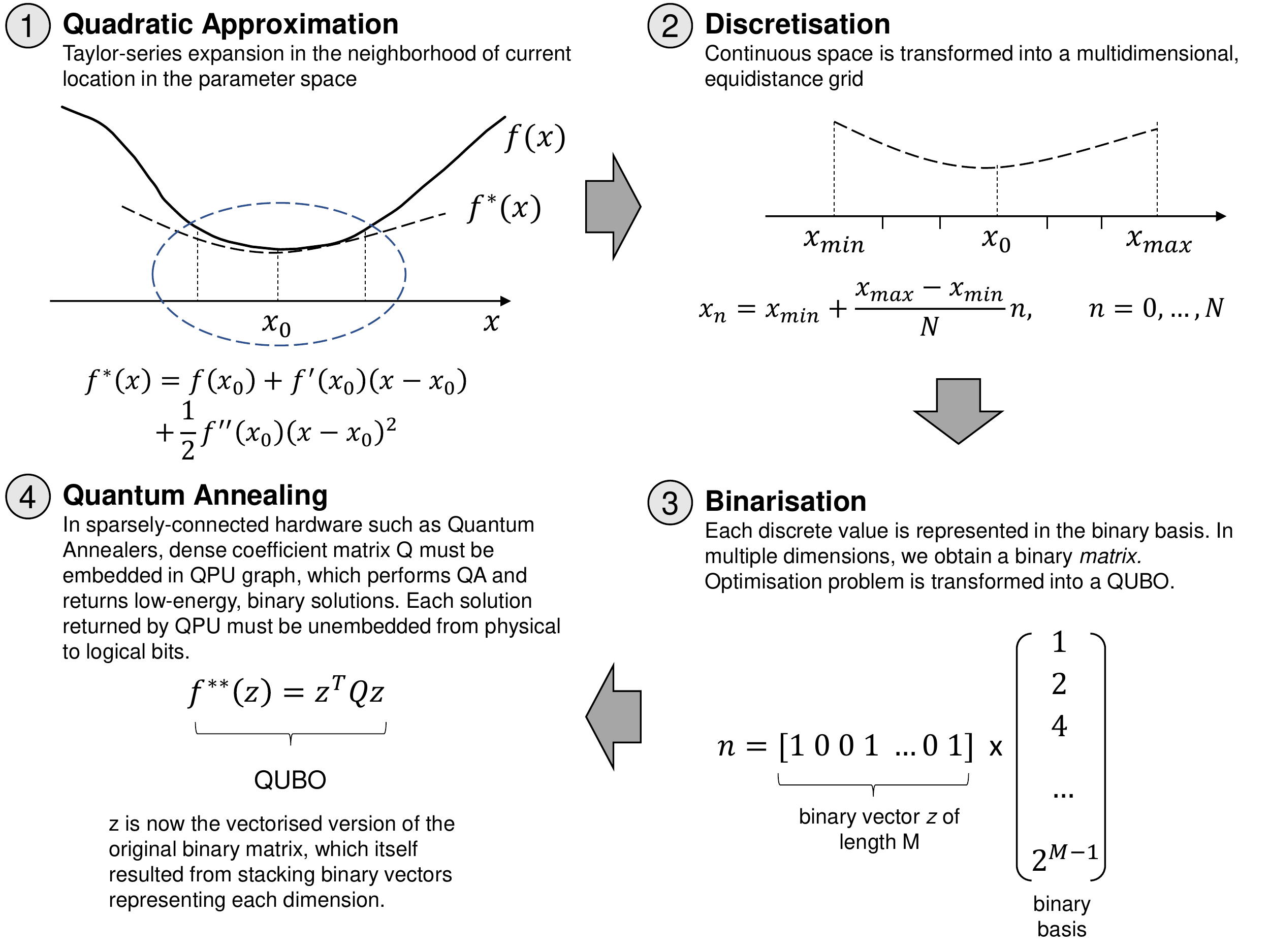}
\caption{Steps involved in solving the TR sub-problem in the QuAnCO algorithm.}
\label{fig-qnlp-overview}
\end{figure}

%\subsection{Converting between QUBO and Ising}

\subsection{TRN vs. BFGS \& CG}

Figure~\ref{fig-supp-trn-vs-bfgs-boxplots} shows comparison plots for TRN vs. BFGS/CG for problem sizes $K=20,200,2000$ and BMP models cone/exponential/Cauchy. While for $K=20$, CG performs as well as - or better - than TRN in terms of final solution quality, the order is reversed for larger problems, i.e., $K=200,2000$.

\begin{figure}[!h]
\centering
\includegraphics[width=12cm]{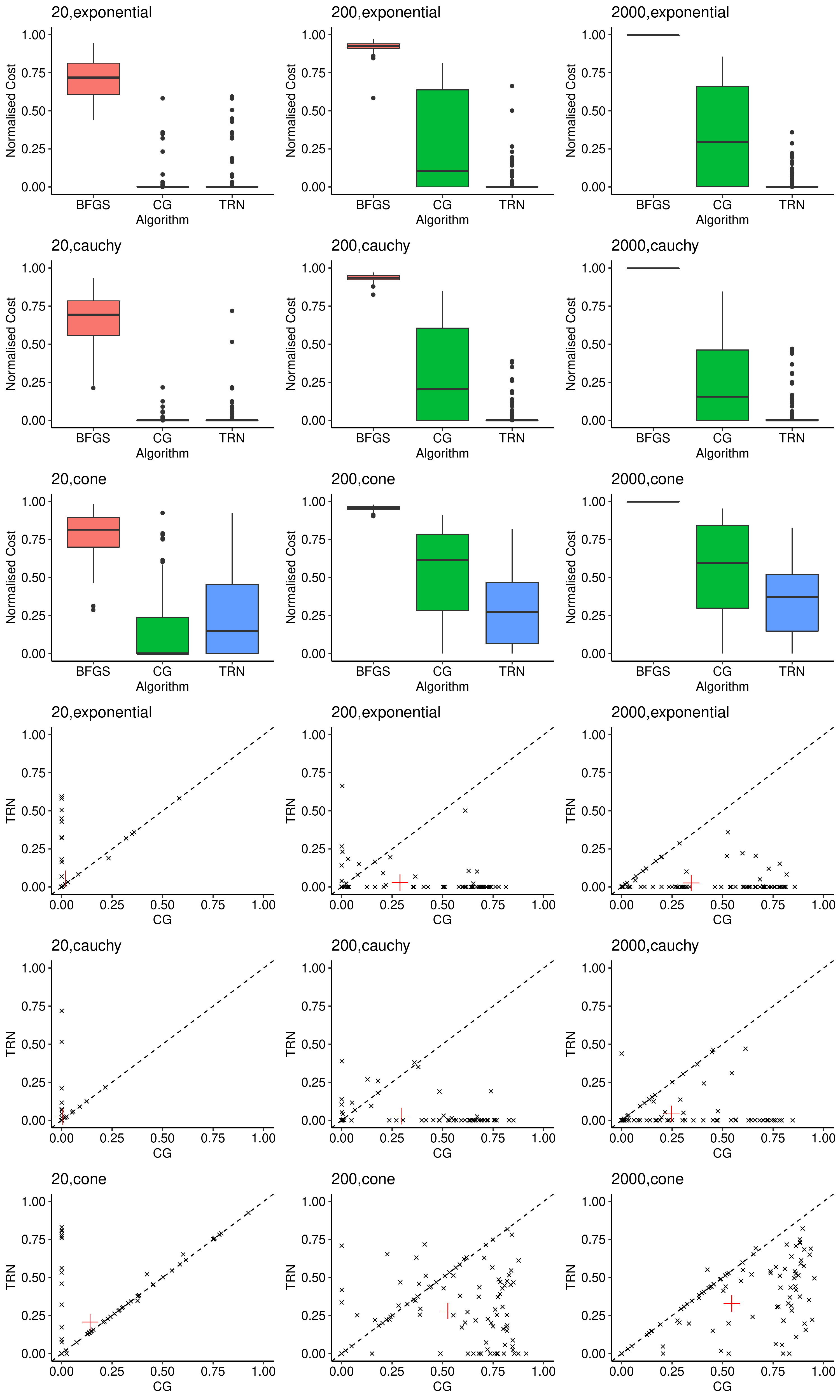}
\caption{Comparing solution quality for TRN vs. BFGS and CG over a range of problem sizes, $K=20,200,2000$, and for three biomethane yield models of cone, exponential and Cauchy. Top three rows show mean normalised cost (after convergence or 100 iterations, whichever comes first). Bottom three rows compare TRN vs. CG, also based on final mean normalised cost numbers across 100 runs.}
\label{fig-supp-trn-vs-bfgs-boxplots}
\end{figure}

\subsection{QuAnCO-Exact vs. TRN}

Figure~\ref{fig-supp-qtro-brute-vs-trn-1} shows comparison plots for QuAnCO-Exact vs. TRN for $K=3,5,7,20$ and $M=1,2,3$ (only $M=1$ for $K=20$), and three biomethane yield models. The QuAnCO-Exact final solution (iter 100) is as good as TRN (exponential and Cauchy) or better (cone).

\begin{figure}[!h]
\centering
\includegraphics[width=14cm]{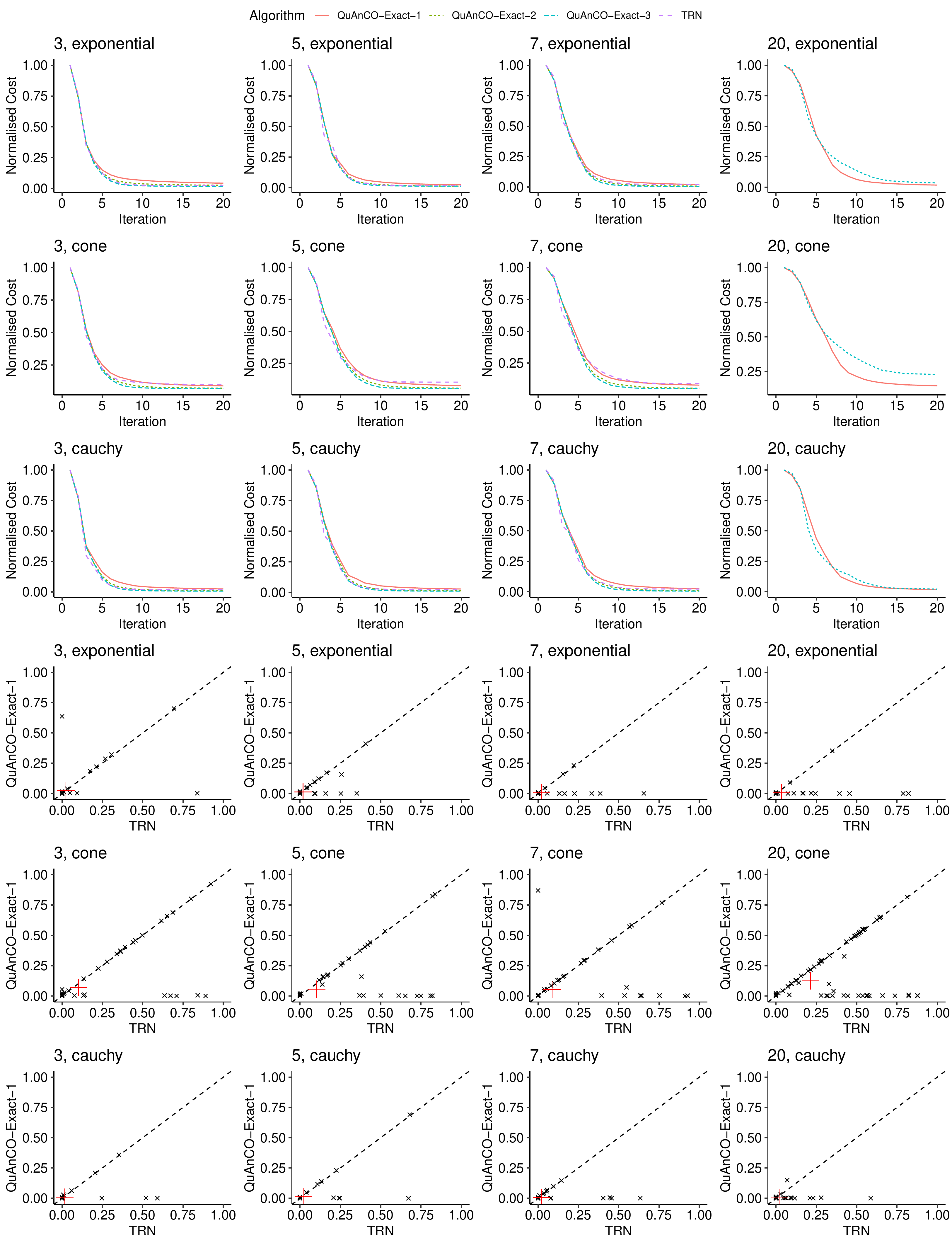}
\caption{Comparing solution quality for QuAnCO-Exact vs. TRN, for problem sizes $K=3,5,7,20$. For first three sizes, we used $M=1,2,3$ bits per dimension, while for $K=20$, only $M=1$ was tried. Top three rows show mean normalised cost vs. iteration number. Bottom three rows show final (iter 100) normalised cost numbers for QuAnCO-Exact-1 vs. TRN.}
\label{fig-supp-qtro-brute-vs-trn-1}
\end{figure}

\subsection{Graph Embedding and Chain Strength}

Figure~\ref{fig-embedding} shows a summary of embedding performance for D-Wave (Chimera, Pegasus) and Hitachi (King) graphs. We make a few observations: 1- Average chain lengths are close for the same graph type, using \textit{clique} and \textit{minorminer} embedding libraries, 2- The Chimera graph requires longer chains, on average, compared to Pegasus, which is due to its sparser connectivity. Similarly, the King graph has a much average chain length than both D-Wave devices, 3- Embedding can become time-consuming, even taking longer than the optimisation itself. However, for fully-connected matrices, embedding can be done once in advance and re-used, and 4- Embedding can fail long before the total number of bits used reaches the maximum graph size, due to connectivity limitations, 4- While \textit{Clique} embedding is much more time-consuming than \textit{minorminer} for the \textit{Pegasus} graph, the reverse is true for the \textit{Chimera} graph.

\begin{figure}[!h]
\centering
\includegraphics[width=\linewidth]{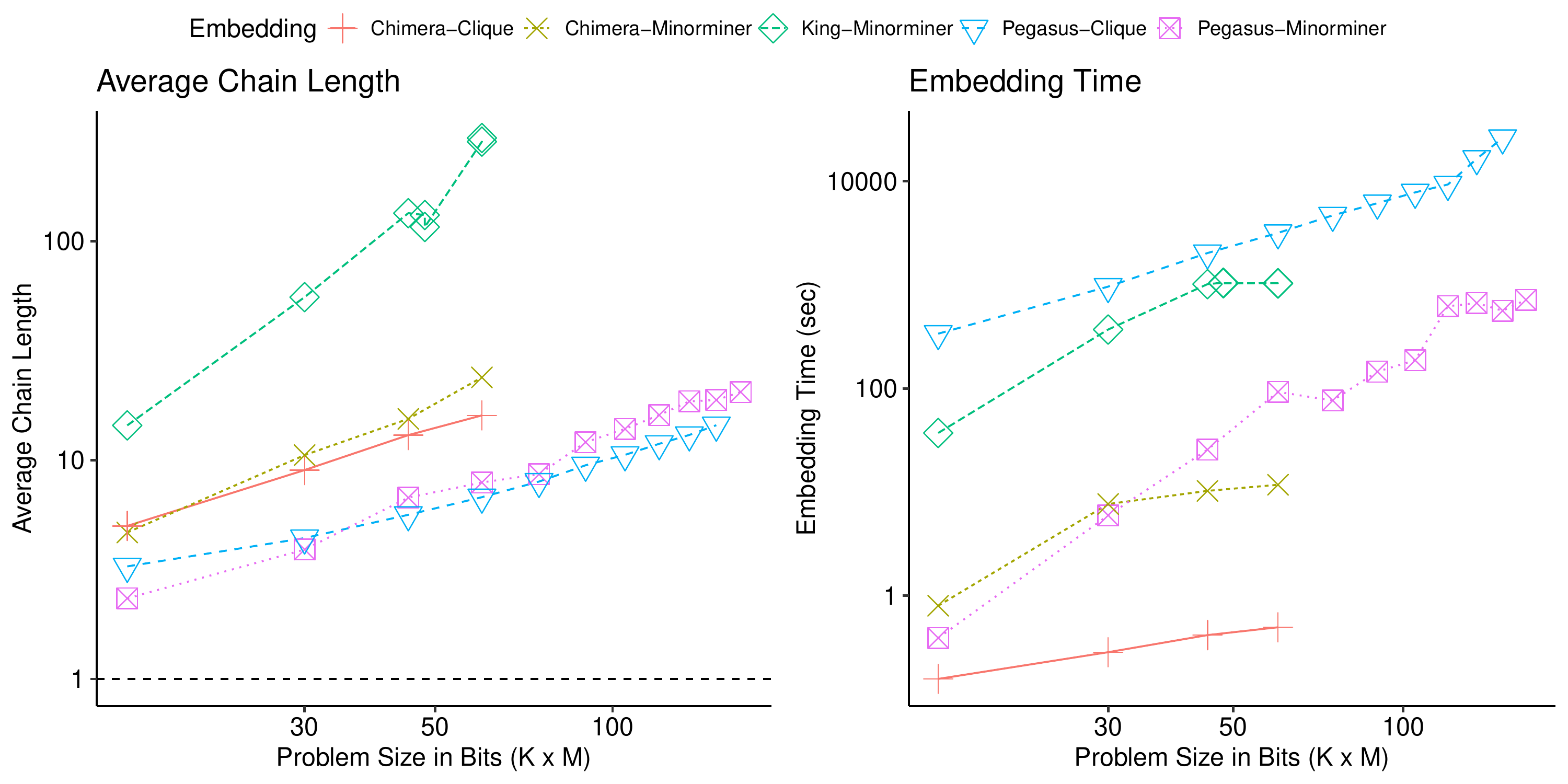}
\caption{Embedding of fully-dense matrices in sparsely-connected graphs. Left: Log-log plot of average chain length vs. problem size in bits. The Chimera and Pegasus graphs correspond to D-Wave's \textit{DW-2000Q} QPU and \textit{Advantage} QPUs, respectively, while the King graph corresponds to Hitachi's CMOS annealer. Horizontal line indicates the ideal case where only one physical bit per logical bit is needed, regardless of problem size. The \textit{minorminer} embedding library from D-Wave can be used on all devices, while the \textit{clique} embedding library is only designed for Chimera and Pegasus graphs. Right: Log-log plot of embedding time (in seconds) for each combination of graph and embedding library, across different problem sizes.}
\label{fig-embedding}
\end{figure}

%\subsection{D-Wave Sensitivity Analysis}

%See Figure~\ref{fig-advantage-sensitivity}.

\subsection{QuAnCO using Quantum-Inspired Ising Solvers}
We tested two quantum-inspired Ising solvers: 1) the Simulated Bifurcation Machine (SBM) from Toshiba~\cite{goto2019combinatorial}, and 2) the CMOS annealer from Hitachi~\cite{yamaoka201520k}.

The SBM in based on simulating adiabatic evolutions of classical nonlinear Hamiltonian systems exhibiting bifurcation phenomena, and takes advantage of the inherent parallelism in system update equations, using GPUs or FPGAs. We tested SBM using the publicly-available GPU-based instance via Amazon AWS~\cite{SBM2022AWS}. This version has a limit of 1000 fully-connected bits, and we tested it with $K=500$ and $M=2$. No embedding and unembedding is needed. Similar to QA and SA, each iteration is based on multiple, independent samples. This number - in our experiments - was auto-selected by the SBM engine, with typical values being in the 300-400 range.

The CMOS annealer is a non-von Neumann computer, transferring data from SRAM directly to calculations. Its implementation on commodity CMOS circuits can provide manufacturing scalability and low power consumption, making it an attractive option for IoT devices. It is available for public testing via a web API~\cite{CMOS2022web}. In addition to the ASIC implementation ($384$ bits), a GPU-based simulator ($512$ bits) has also been provided for public access, as of this writing. Both versions are connected in King's graph pattern, thus requiring embedding and unembedding, similar to D-Wave's QAs. We tested $K=24$, $M=2$ using both GPU and ASIC options, referred to as device numbers 4 and 5, respectively. For the ASIC version, due to limited dynamic range of the device, the weights must be rescaled to fall between -7 and +7.

Figure~\ref{fig-suppmat-sbm} shows the results. QuAnCO-SBM-2 has a similar - and slightly better - performance than QuAnCO-SA-2, both clearly outperforming TRN for the cone problem and $K=500$. As the scatter plot confirms (top right panel), the QuAnCO-SBM-2 advantage is nearly risk-free. On the other hand, the ASIC version of CMOS annealer significantly underperforms TRN for $K=24$, while the GPU version performs nearly identical to TRN.

\begin{figure}[!h]
\centering
\includegraphics[width=\linewidth]{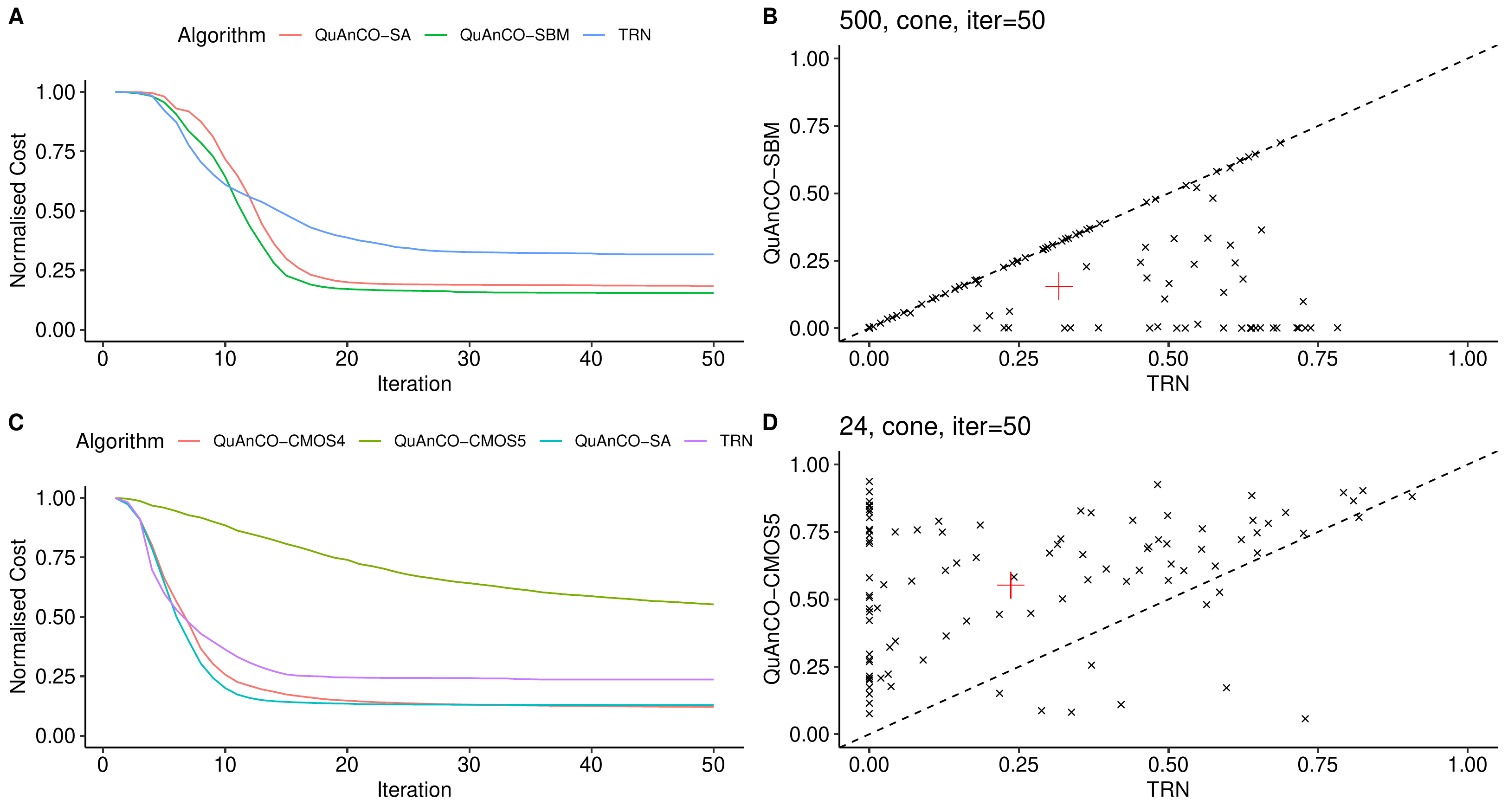}
\caption{Performance of QuAnCO using quantum-inspired Ising solvers, compared to TRN and QuAnCO-SA. Top Row: Simulated Bifurcation Machine from Toshiba. Bottom Row: CMOS annealer from Hitachi (device 4: GPU, device 5: ASIC). All QuAnCO algorithms use $M=2$ bits per dimension. The cone model of biomethane yield was used. Experiments lasted 50 iterations (or until convergence for TRN).}
\label{fig-suppmat-sbm}
\end{figure}

\subsection{Handling Nonlinear Constraints}\label{subsec-suppmat-constraints}
A combination of penalty terms and change-of-variables can be used to absorb nonlinear equality and inequality constraints. Below we describe the outline of such a strategy, starting with equality constraints.

\textbf{Equality Constraint} \quad Consider the following optimisation problem with a nonlinear equality constraint:
\begin{equation}\label{eq-optim-nonlinear-const-equality}
\begin{aligned}
\min_{\xx} \quad & f(\xx) \\
\textrm{s.t.} \quad & g(\xx) = a.
\end{aligned}
\end{equation}
We replace the equality constraint with a quadratic penalty term:
\begin{equation}\label{eq-optim-nonlinear-const-equality-2}
\begin{aligned}
\min_{\xx} \quad & f(\xx) + \lambda \, (g(\xx) - a)^2.
\end{aligned}
\end{equation}
The larger $\lambda$ is, the stricter the equality constraint is enforced. It is typical to start with a small $\lambda$, and use the solution found as a starting point to increase $\lambda$ in the next round, until the equality constrained is satisfied to the desired level of precision. Handling multiple equality constraints requires introducing several such penalty terms, each with their corresponding $\lambda$.

\textbf{Inequality Constraint} \quad Let's consider a single box constraint:
\begin{equation}\label{eq-optim-nonlinear-const-inequality}
\begin{aligned}
\min_{\xx \in \R^K} \quad & f(\xx) \\
\textrm{s.t.} \quad & a \leq g(\xx) \leq b.
\end{aligned}
\end{equation}
We introduce a new variable, $z$, along with a penalty term to enforce the resulting equality constraint, which results in a bound-constrained optimisation problem:
\begin{equation}\label{eq-optim-nonlinear-const-inequality-2}
\begin{aligned}
\min_{\xx' \in \R^{K+1}} \quad & f(\xx) + \lambda \, (g(\xx) - z)^2 \\
\textrm{s.t.} \quad & a \leq z \leq b.% \\
%& g(\xx) = z.
\end{aligned}
\end{equation}
In the above, $\xx'$ is a concatenation of $\xx$ and $z$, i.e., $\xx'=[\xx^\top \, z]^\top$. We can now use the nonlinear transformation technique described in the main text (Section~\ref{subsec-trn-bound}) to handle the bound constraint. Note that for the first $K$ elements of $\xx'$, corresponding to $\xx$, we do not need any nonlinearity used, i.e., $\eta_k(\xx) = 1, \, \forall k=1,\cdots,K$. Further inequality constraints can be added in a similar fashion, and this can be combined with the equality-constraint approach above.

Full mathematical derivations, loop strategies needed for dialing $\lambda$'s up/down, and experiments on effectiveness of this approach are all topics for future research.

\section{Data Availability}

The data that support the findings of this study are available from the corresponding author upon reasonable request, and subject to Nature Energy's approval.

\section{Code Availability}

Computer code implementing the QuAnCO algorithm is available from the corresponding author upon reasonable request. %Authors intend to release an open-source implementation of QuAnCO for public use in the near future.

\bibliography{main}

%\noindent LaTeX formats citations and references automatically using the bibliography records in your .bib file, which you can edit via the project menu. Use the cite command for an inline citation, e.g.  \cite{Hao:gidmaps:2014}.

%For data citations of datasets uploaded to e.g. \emph{figshare}, please use the \verb|howpublished| option in the bib entry to specify the platform and the link, as in the \verb|Hao:gidmaps:2014| example in the sample bibliography file.

%\section*{Acknowledgements (not compulsory)}
%Acknowledgements should be brief, and should not include thanks to anonymous referees and editors, or effusive comments. Grant or contribution numbers may be acknowledged.

\section{Author contributions statement}

%Must include all authors, identified by initials, for example:
%A.A. conceived the experiment(s),  A.A. and B.A. conducted the experiment(s), C.A. and D.A. analysed the results.  All authors reviewed the manuscript. 

M.T.A.S developed the mathematical framework and contributed to software development. V.B.J. collected biogas data and parameters. M.J.  supplied the biogas production assumptions and constraints. A.S.M. contributed to problem definition and software development. All authors reviewed the manuscript.

%\section*{Additional information}

%To include, in this order: \textbf{Accession codes} (where applicable); \textbf{Competing interests} (mandatory statement). 

%The corresponding author is responsible for submitting a \href{http://www.nature.com/srep/policies/index.html#competing}{competing interests statement} on behalf of all authors of the paper. This statement must be included in the submitted article file.

%\begin{figure}[ht]
%\centering
%\includegraphics[width=\linewidth]{stream}
%\caption{Legend (350 words max). Example legend text.}
%\label{fig:stream}
%\end{figure}

%\begin{table}[ht]
%\centering
%\begin{tabular}{|l|l|l|}
%\hline
%Condition & n & p \\
%\hline
%A & 5 & 0.1 \\
%\hline
%B & 10 & 0.01 \\
%\hline
%\end{tabular}
%\caption{\label{tab:example}Legend (350 words max). Example legend text.}
%\end{table}

%Figures and tables can be referenced in LaTeX using the ref command, e.g. Figure \ref{fig:stream} and Table \ref{tab:example}.

\end{document}